\documentclass[a4paper,11pt]{article}
\usepackage{fullpage}
\usepackage{rotating}
\usepackage{mathpazo}
\usepackage{tablefootnote}
\usepackage{graphics}
\usepackage{setspace}
\usepackage{booktabs}
\usepackage[scale=2]{ccicons}
\usepackage{stmaryrd}

\usepackage{float}
\usepackage{hyperref}
\usepackage{color}
\usepackage{amsthm}
\usepackage[absolute,overlay]{textpos}
\usepackage{pgfplots}

\usepackage{tikz}
\tikzstyle{vertex}=[circle, draw, inner sep=0pt, minimum size=6pt]

\usepackage{caption}
\usetikzlibrary{patterns,positioning,decorations.pathmorphing,decorations.pathreplacing}
\usetikzlibrary{arrows,automata,snakes}
\usepackage{graphicx}
\usepackage{amsmath}
\usepackage{amssymb}
\usepackage{diagbox}
\usepackage{caption}
\usepackage{hyperref}
\hypersetup{colorlinks=true,linkcolor=blue,citecolor=blue}
\usepackage{array,booktabs,arydshln,xcolor}

\newcolumntype{M}[1]{>{\centering\arraybackslash}m{#1}}
\makeatletter
    \newcommand{\thickhline}{%
        \noalign {\ifnum 0=`}\fi \hrule height 1pt
        \futurelet \reserved@a \@xhline
    }
    \newcolumntype{"}{@{\vrule width 1pt}}
    \makeatother

\newtheorem{theorem}{Theorem}

\newtheorem{lemma}[theorem]{Lemma}

\mathchardef\mhyphen="2D
\newcommand\MSD{\text{SAL}}
\newcommand\CSD{\text{CSD}}
\newcommand\ST{\text{ST}}

\title{\textbf{An Efficient Representation for\\ Filtrations of Simplicial Complexes}}

\date{}
\author{Jean-Daniel Boissonnat\thanks{This work was partially supported by the Advanced Grant of the European Research Council GUDHI (Geometric Understanding in Higher Dimensions). }\\
INRIA Sophia Antipolis - M\'{e}diterran\'{e}e,
France. \\
\texttt{Jean-Daniel.Boissonnat@inria.fr}. \vspace{0.5cm}
\and 
Karthik C.\ S.\footnote{ This work was partially supported by Irit Dinur's ERC-StG grant number 239985.  } \\
Department of Computer Science and Applied Mathematics,\\
 Weizmann Institute of Science,
  Israel.\\
   \texttt{karthik.srikanta@weizmann.ac.il}.
}

\begin{document}
\maketitle
\vspace{-1cm}
\begin{abstract}
A filtration over a simplicial complex $K$ is an ordering of the simplices of $K$ such that all prefixes in the ordering are subcomplexes of $K$. Filtrations are at the core of Persistent Homology, a major tool in Topological Data Analysis. In order to represent the filtration of a simplicial complex, the entire filtration can be appended to any data structure that explicitly stores all the simplices of the complex  such as the Hasse diagram  or the recently introduced Simplex Tree [Algorithmica '14]. However, with the popularity of various computational methods that need to handle simplicial complexes, and with the rapidly increasing size of the complexes, the task of finding a compact data structure that can still support efficient queries is of great interest. \\

This direction has been recently pursued for the case of maintaining simplicial complexes. For instance, Boissonnat et al.\ [Algorithmica '17] considered storing the simplices that are maximal with respect to inclusion and Attali et al.\ [IJCGA '12] considered storing the simplices that block the expansion of the complex. 
Nevertheless, so far there has been no data structure that compactly stores the \emph{filtration} of a simplicial complex, while also allowing the efficient implementation of basic operations on the complex.\\

In this paper, we propose a new data structure called the Critical Simplex Diagram (CSD) which is a variant of the Simplex Array List (SAL) [Algorithmica '17]. Our data structure allows one to store in a compact way the filtration of a simplicial complex, and allows for the efficient implementation of a large range of basic operations. Moreover, we prove that our data structure is essentially optimal with respect to the requisite storage space. Finally, we show that the CSD representation admits fast construction algorithms for Flag complexes and  relaxed Delaunay complexes.
\end{abstract}
\clearpage
\section{Introduction}
{Simplicial complexes are the prime objects to represent topological spaces. The notion of filtration of a simplicial complex has been introduced to allow the representation of topological spaces at various scales and to distinguish between the true features of a space and artifacts arising from bad sampling, noise, or a particular choice of parameters~\cite{EH10}.  The most popular filtrations are nested sequences of increasing simplicial complexes but more advanced types of filtrations have been studied where consecutive complexes are mapped using more general simplicial maps~\cite{DFW14}. Filtrations are at the core of Persistent Homology, a major tool in the emerging field of Topological Data Analysis. 
}

A central question in Computational Topology and Topological Data Analysis is thus to represent simplicial complexes and filtrations efficiently.  The most common representation of simplicial complexes uses the Hasse diagram of the complex that has one node per simplex and an edge between any pair of incident simplices whose dimensions differ by one.  A more compact data structure, called Simplex Tree (ST), was proposed recently by Boissonnat and Maria \cite{SimplexTree}. The nodes of both the Hasse diagram and ST are in bijection with the simplices (of all dimensions) of the simplicial complex. In this way, they explicitly store all the simplices of the complex and it is easy to attach information to each simplex (such as a filtration value).  In particular, they allow one to store in an easy way the filtration of complexes. 

However, such data structures are typically very big, and they are not sensitive to the underlying structure of the complexes. This motivated the design of more compact data structures that represent only a sufficient subset of the simplices. A first idea is to store the 1-skeleton of the complex
together with a set of blockers that prevent the expansion of the complex~\cite{DataStructure3}.  A dual idea is to store only the simplices that are maximal with respect to inclusion. Following this last idea,
Boissonnat et al.\ \cite{BKT15} introduced a new data structure, called the Simplex Array List, which was the first data structure whose size  and query time  were sensitive to the geometry of the simplicial complex.  
SAL was shown to outperform  ST for a large class of simplicial complexes.

Although very efficient, SAL, as well as  data structures that do not  explicitly store all the simplices of a complex, makes the representation of filtrations problematic, and in the case of SAL, impossible. In this paper, we introduce a new data structure called Critical Simplex Diagram (CSD) which has some similarity with SAL. CSD only stores the critical simplices, i.e., those simplices all of whose cofaces have a higher filtration value, and in this paper, we overcome the problems arising due to the implicit representation of  simplicial complexes, by showing that the basic operations on simplicial complexes can be performed efficiently using CSD.

\subsection{Our Contribution}
At a high level, our main contribution through this paper is to develop a new perspective for the design of data structures representing simplicial complexes associated with a filtration. Previous data structures such as the Simplex Tree interpreted a simplicial complex as a set of strings defined over the set of labels of its vertices and the filtration values as keys associated with each string. When a simplicial complex is perceived this way, a trie is indeed a natural data structure to represent the complex.  However, this way of representing simplicial complexes doesn't make use of the fact that simplicial complexes are not arbitrary sets of strings but are constrained by a lot of combinatorial structure. In particular, simplicial complexes are closed under subsets and also (standard) filtrations are monotone  functions\footnote{By monotone function we mean that the filtration value of a coface of any simplex is at least the filtration value of the simplex.}. 

We exploit this structure by viewing a filtered simplicial complex with a filtration range of size $t$ as a monotone function from $\{0,1\}^{|V|}$ to $\{0,1,\ldots,t\}$, where $V$ is the vertex set. We note that if a simplex is mapped to $t$ then, the simplex is understood to be not in the complex and if not, the mapping is taken to correspond to the filtration value of the simplex. In light of this viewpoint, we propose a data structure (CSD) which stores only the critical elements in the domain, i.e. those elements all of whose supersets (cofaces in the complex) are mapped to a strictly larger value. As a result, we are able to store the data regarding a simplicial complex more compactly. More concretely, we have the following result.

\begin{theorem}\label{main}
Let $K$ be a $d$-dimensional simplicial complex. Let $\kappa$ be the number of critical simplices in the complex. The data structure CSD representing $K$ admits the following properties:
\begin{enumerate}
\item The size of CSD is  $\tilde{\mathcal{O}}(\kappa d)$.
\item The cost of basic operations (such as membership, insertion, removal, elementary collapse, etc.) through the CSD representation is $\tilde{\mathcal{O}}((\kappa \cdot d)^2)$. 
\end{enumerate}
\end{theorem}

The proof of the above two properties follows from the discussions in Section~\ref{sizeofCSD} and Section~\ref{sec:operations} respectively. We would like to point out here that while the cost of static operations such as membership is only $\tilde{\mathcal{O}}(d)$ for the Simplex Tree, to perform any dynamic operation such as insertion or removal, the Simplex Tree requires $\text{exp}(d)$ time. 

As a direct consequence of representing a simplicial complex only through the critical simplices, 
the construction of any simplicial complex with filtration, will be very efficient through CSD, simply because we have to build a smaller data structure  as  compared to the existing data structures. This is shown for flag complexes and relaxed Delaunay complexes.

\section{Preliminaries}\label{Preliminaries}

A simplicial complex $K$ is defined over a (finite) vertex set $V$ whose elements are called the vertices of $K$ and is a set of non-empty subsets of $V$ that is required to satisfy the following two conditions:
\begin{enumerate}
\item $p\in V\Rightarrow \{p\}\in K$
\item $\sigma\in K, \tau\subseteq\sigma\Rightarrow\tau\in K$
\end{enumerate}

Each element $\sigma\in K$ is called a simplex or a face of $K$. {The dimension $d_{\sigma}$ of $\sigma$ is equal to its number of vertices minus 1. 
A simplex of dimension $d$ is also called a $d$-simplex for short.}
  The dimension of the simplicial complex $K$ is the largest $d$ such that it contains a $d$-simplex.

A \emph{face} of a simplex $\sigma = \{p_0 ,..., p_s \}$ is a simplex whose vertices form a subset of $\{p_0 ,..., p_s \}$. A proper face is a face different from $\sigma$ and the facets of $\sigma$ are its proper faces of maximal dimension. A simplex $\tau\in K$ having $\sigma$ as a face is called a \emph{coface} of $\sigma$. {In some places in the paper, we will not specify if a face or coface is proper, as it should be clear from the context.}

A maximal simplex of a simplicial complex is a simplex that has no cofaces. 
A simplicial complex is pure if all its maximal simplices are of the
same dimension. Also, a \emph{free pair} is defined as a pair of simplices
$(\tau,\sigma)$ in $K$ where $\tau$ is the only coface of $\sigma$.

In Figure~\ref{fig:SimplicialComplexExample} we see a three dimensional simplicial complex on the vertex set $\{1,2,3,4,5,6\}$. This complex has two maximal simplices: the tetrahedron $[1234]$ and  the triangle $[356]$. We use this complex as an example throughout the paper.  

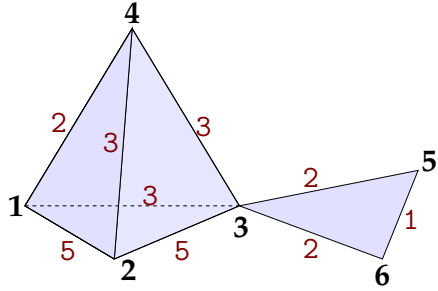
\begin{figure}[!h]
\centering
\resizebox{6cm}{!}{
\begin{tikzpicture}[-,>=stealth',shorten >=0.5pt,auto,node distance=2cm,
 thick,main node/.style={circle,fill=blue!10,draw,font=\sffamily\large\bfseries}]
\draw [fill=blue!10!white] (0,0)--(3,5) --(6,0)-- (2.5,-1.5) --cycle; 
\draw [fill=blue!15!white] (2.5,-1.5)--(3,5) --(0,0) --cycle; 
\draw [fill=blue!12!white] (2.5,-1.5)--(3,5) --(0,0) --cycle; 
\draw [fill=blue!12!white] (6,0)--(11,1) --(10,-1.5) --cycle; 
\Huge
\draw [dashed] (0,0) -- (6,0);
\draw [-] (0,0) -- (3,5) -- (6,0);
\draw [-] (0,0) -- (2.5,-1.5) -- (6,0);
\node at (-0.25,0) {\textbf{1}};
\node at (6.05,-0.6) {\textbf{3}};
\node at (3,5.45) {\textbf{4}};
\node at (11.3,1.25) {\textbf{5}};
\node at (10,-1.95) {\textbf{6}};
\node at (2.9,-1.8) {\textbf{2}};
\node at (1.2,-1.25) {\color{red!50!black}\texttt{5}};
\node at (4.4,-1.25) {\color{red!50!black}\texttt{5}};
\node at (0.95,2.35) {\color{red!50!black}\texttt{2}};
\node at (5,2.25) {\color{red!50!black}\texttt{3}};
\node at (3.5,0.3) {\color{red!50!black}\texttt{3}};
\node at (2.4,1.8) {\color{red!50!black}\texttt{3}};
\node at (8,0.8) {\color{red!50!black}\texttt{2}};
\node at (8,-1.2) {\color{red!50!black}\texttt{2}};
\node at (10.8,-0.4) {\color{red!50!black}\texttt{1}};

\end{tikzpicture}
}
\caption{A simplicial complex with two maximal simplices~: tetrahedron $[1234]$ and triangle $[356]$. The filtration value of all vertices is $0$. Filtration values of edges are shown in the figure. The filtration value of a higher dimensional simplex is the largest filtration value of its edges.}
\label{fig:SimplicialComplexExample}
\end{figure}

We adopt the following notation: 
$[t]:=\{1,\ldots ,t\}$ and $\llbracket t\rrbracket=\{0,1,\ldots ,t\}$. A {\em filtration} of a complex is a function $f : K \to \mathbb{R}$ satisfying $f(\tau ) \le f(\sigma)$
whenever $\tau\subseteq\sigma$ \cite{EH10}. Moreover, we will assume that the  filtration values range over $\llbracket t\rrbracket$. We say that a simplex $\sigma\in K$ is a \emph{critical simplex} if for all cofaces $\tau$ of $\sigma$ we have $f(\sigma)<f(\tau)$. For example, the critical simplices in the example described in Figure~\ref{fig:SimplicialComplexExample} are all the vertices, the edges $[56], [14], $ and $[24]$, the triangles $[356]$ and $[134]$, and the tetrahedron $[1234]$.

\paragraph{Notations.} In this paper, the class of simplicial complexes of
$n$ vertices and dimension $d$  with $k$ maximal simplices out of the $m$ simplices in the complex is denoted by
${\cal K}(n,d,k,m)$,  and $K$ denotes a simplicial complex in $ {\cal
  K} (n,d,k,m)$. 
 
\subsection{Lower Bounds}  \label{Lower Bounds}
  Boissonnat et al.\ proved the following lower bound on the space needed to represent simplicial
complexes \cite{BKT15}.
\begin{theorem}
\label{Lowerbound}
\cite{BKT15}
Consider the class of all $d$-dimensional simplicial complexes with $n$ vertices containing $k$ maximal simplices, where $d\ge2$ and $k\ge n+1$, and consider 
any data structure that can represent the simplicial complexes of this class. Such a data structure requires $\log{\binom{\binom{n/2}{d+1}}{k-n}}$ bits to be stored. For any constant  $\varepsilon\in (0,1)$ and for  $\frac{2}{\varepsilon}n\le k\le n^{(1-\varepsilon)d}$, $d\le n^{\varepsilon/3}$,  the bound becomes $\Omega(kd\log n)$.
\end{theorem}

We prove now a lower bound on the representation of filtrations of simplicial complexes.

\begin{lemma}\label{FiltrationLowerbound}
Let $\beta=\left\lfloor\frac{t+1}{d+1}\right\rfloor$ be greater than 1. For any simplicial complex $K$ of dimension $d$ containing $m$ simplices, the number of distinct filtrations $f:K\to\llbracket t\rrbracket$ is at least $\beta^{m}$. If $\beta>(d+1)^{\delta}$ for some constant $\delta>0$ then, any data structure that can represent filtrations of the class of all $d$-dimensional simplicial complexes containing $m$ simplices requires $\Omega(m\log t)$ bits to be stored. 
\end{lemma}
\begin{proof}
Let us fix a simplicial complex $K$ of dimension $d$ containing $m$ simplices.  We will now build functions $f_i:K\to\llbracket t\rrbracket$. For every $i\in \left\llbracket{\beta}^m-1\right\rrbracket$, let $b(i)$ be the representation of $i$ as an $m$ digit number in base $\beta$ and let $b(i)_j$ be the $j^{\text{th}}$ digit of $b(i)$. Let $g$ be a bijection from $K$ to $[m]$. We define $f_i(\sigma)=d_{\sigma}\cdot\beta +b(i)_{g(\sigma)}+1$. We note that all the $f_i$s are distinct functions as for any two distinct numbers $i$ and $j$ in $\left\llbracket{\beta}^m-1\right\rrbracket$, we have that $b(i)\neq b(j)$. Finally, we note that each of the $f_i$s is a  filtration of $K$. This is because, for any two simplices $\tau,\sigma\in K$, such that $\tau\subset\sigma$ (i.e., $d_\tau<d_\sigma$), and any $i\in \left\llbracket{\beta}^m-1\right\rrbracket$, we have that $f_i(\tau)\le (d_{\tau}+1)\cdot\beta < d_{\sigma}\cdot\beta+1\le f_i(\sigma)$.

It follows that there are at least $\beta^{m}$ distinct filtrations of $K$. By the pigeonhole principle, we have that  any data structure that can represent  filtrations of $K$  requires $\log\left(\beta^{m}\right)$ bits. It follows that if $\beta>(d+1)^{\delta}$ for some constant $\delta>0$ then $\beta>\frac{(t+1)^\delta}{(\beta-1)^\delta}$. Thus, 
any data structure that can represent filtrations of $K$ requires at least $\frac{\delta}{1+\delta}\cdot m\log (t+1)=\Omega(m\log t)$ bits to be stored. 
\end{proof}

Even if $\left\lfloor\frac{t}{d}\right\rfloor \le 1 $, we can show that any data structure that can represent $d$-dimensional simplicial complexes  containing $m$ simplices with filtration range $\llbracket t\rrbracket$ requires $\Omega(\frac{m\sqrt{t}}{d}\log t)$ bits. This can be shown by modifying the above proof as follows. Let $S_j$ be the set of all simplices of dimension $j-1$. We identify a subset $D$ of $[d]$ of size $\sqrt{t}$, such that $\sum_{j\in D}|S_{j}|$ is at least $\frac{m\sqrt{t}}{d}$. Therefore, for every set $S_j$, $j\in D$, we can associate $\sqrt{t}$ distinct filtration values, which leads to the lower bound.

The lower bound in Lemma~\ref{FiltrationLowerbound} is not sensitive to the number of critical simplices, and intuitively, any lower bound on the size of data structures storing complexes with filtrations needs to capture the number of critical simplices as a parameter. We adapt the proof of Theorem~\ref{Lowerbound} and combine it with the ideas from the proof of Lemma~\ref{FiltrationLowerbound}, to obtain the following lower bound. {Intuitively, the theorem says that given $\kappa$ critical simplices, with no other restrictions, it is not possible to find a representation shorter than an enumeration of these critical simplices}, each of these critical simplices being represented with $d\log n + \log t$ bits. 

\begin{theorem}\label{CSDLowerBound}
Consider the class of all simplicial complexes on $n$ vertices of dimension $d$, associated with a filtration over the range of $\llbracket t\rrbracket$, such that the number of critical simplices is $\kappa$, where $d\ge2$ and $\kappa\ge n+1$, and consider any data structure that can represent the simplicial complexes of this class. Such a data structure requires $\log\left(\binom{\binom{n/2}{d+1}}{\kappa-n}t^{\kappa-n}\right)$ bits to be stored. For any constant  $\varepsilon\in (0,1)$ and for  $\frac{2}{\varepsilon}n\le \kappa\le n^{(1-\varepsilon)d}$ and $d\le n^{\varepsilon/3}$,  the bound becomes $\Omega(\kappa (d\log n+\log t))$.
\end{theorem}
\begin{proof} The proof of the first statement  is by contradiction. 
Let us define $h=\kappa-n\ge 1$ and suppose that there exists a data structure that can be stored using only $s<\log \alpha \stackrel{{\rm def}}{=}\log\left(\binom{\binom{n/2}{d+1}}{\kappa-n}t^{h}\right)$ bits. 
We will construct $\alpha$ simplicial complexes (associated with a filtration), all with the same set  $P$ of $n$ vertices,  the same dimension $d$, with exactly $\kappa$ maximal simplices, and with a filtration over the range of $\llbracket t \rrbracket$. By the pigeonhole principle, two different filtered simplicial complexes\footnote{i.e., two complexes which either differ on the simplices contained on the complex or on the filtration value of a simplex contained in both complexes.}, say $K$ and $K^\prime$, are encoded by the same word. So any algorithm will give the same answer for $K$ and $K^\prime$. But, by the construction of these complexes, there is either a simplex which is in $K$ and not in $K^\prime$ or there is a simplex whose filtration value in $K$ is different from the simplex's filtration value in $K^\prime$. This leads to a contradiction. 

The simplicial complexes and their associated filtration are constructed as follows. Let $P'\subset P$ be a subset of cardinality $n/2$, and consider the set of all possible simplicial complexes of dimension $d$ with vertices in $P'$ that contain  $h$ critical simplices. We further assume that all critical simplices have dimension $d$ exactly. These complexes are $\beta=\binom{\binom{n/2}{d+1}}{h}$ in number, since the total number of maximal $d$ dimensional simplices is $\binom{n/2}{d+1}$ and we choose $h$ of them. Let us call them $\Gamma_1,\ldots,\Gamma_\beta$.  We now extend each $\Gamma_i$ so as to obtain a simplicial complex whose vertex set is $P$ and has exactly $\kappa$ critical simplices. The critical simplices will consist of  the $h$ maximal simplices of dimension $d$ already constructed (whose filtration value is set to one of the values in $[t]$) plus a number of maximal simplices of dimension 1 (whose filtration value is set to 0).  The set of vertices of $\Gamma_i$, ${\rm vert}(\Gamma_i)$,  may be a strict subset of $P'$. Let its cardinality be $\frac{n}{2}-r_i$  and observe that $0\leq r_i<\frac{n}{2}$. Consider now the complete graph on the $\frac{n}{2}+r_i$ vertices of $P\setminus {\rm vert}(\Gamma_i)$.  Any spanning tree of this graph gives $\frac{n}{2}+r_i-1$ edges and we arbitrarily choose $\frac{n}{2}-r_i+1$ edges from the remaining edges of the graph to obtain $n$ distinct edges spanning over the  vertices of $P\setminus {\rm vert}(\Gamma_i)$.  We have thus constructed a 1--dimensional simplicial complex $K_i$ on the $\frac{n}{2}+r_i$ vertices of $P\setminus {\rm vert}(\Gamma_i)$ with exactly $n$ maximal simplices. Finally, we define the complex $\Lambda_i=\Gamma_i\cup K_i$ that has $P$ as its vertex set, dimension $d$, and  $\kappa$ maximal simplices which are also the critical simplices. The filtration value of any simplex which is not maximal is defined to be the minimum of the filtration values of its cofaces in the complex.  
The set of $\Lambda_i$, $i=1, \cdots ,\beta$, where for each complex we associate $t^h$ different filtrations is the set of simplicial complexes (associated with a filtration) that we were  looking for.

The second statement in the theorem is proved through the following computation:

\allowdisplaybreaks
\begin{align*}
\log\left(\dbinom{\dbinom{n/2}{d+1}}{\kappa-n}t^{\kappa-n}\right)\ge &\ \log{\left(\frac{n^{(d+1)(\kappa-n)}}{2^{(d+1)(\kappa-n)}(d+1)^{(d+1)(\kappa-n)}(\kappa-n)^{(\kappa-n)}}\right)}+(\kappa-n)\log t \\
= &\ (d+1)(\kappa-n)\log n -(d+1)(\kappa-n)- (d+1)(\kappa-n)\log (d+1)\\
 &\ \phantom{{}(d+1)(\kappa-n)\log n -(d+1)} - (\kappa-n)\log (\kappa-n) +(\kappa-n)\log t\\
> &\ (d+1)(\kappa-n)\log n -3(d+1)(\kappa-n)-(d+1)(\kappa-n)\log d \\
 &\ \phantom{{}(d+1)(\kappa-n)\log n -(d+1)\log n } - (\kappa-n)\log \kappa +(\kappa-n)\log t\\
\ge &\ (d+1)(\kappa-n)(\log n -3-\log d) - (\kappa-n)(1-\varepsilon)d\log n \\
&\ \phantom{{}(d+1)(\kappa-n)(\log n -3-\log d) - (\kappa-n)(1-\varepsilon)d}  +(\kappa-n)\log t\\
\ge &\ d\varepsilon(\kappa-n)\log n + (\kappa-n)\log n -(d+1)(\kappa-n)\left(3+\frac{\varepsilon}{3}\log n\right)\\
&\ \phantom{{}(d+1)(\kappa-n)(\log n -3-\log d) - (\kappa-n)(1-\varepsilon)d}  +(\kappa-n)\log t\\
\ge &\ \frac{2\varepsilon}{3}\left(1-\frac{\varepsilon}{2}\right)\kappa d\log n  +\left(1-\frac{\varepsilon}{2}\right)\kappa\log t+\left(1-\frac{\varepsilon}{2}\right)\kappa\log n \\
&\ \phantom{{}(\kappa-n)\log n -(d+1)}-3d\left(1-\frac{\varepsilon}{2}\right)\kappa -\left(1-\frac{\varepsilon}{2}\right)\kappa\left(3+\frac{\varepsilon}{3}\log n\right) \\
= &\ \Omega(\kappa (d\log n+\log t))
\end{align*}

We note that in the above computation, the first inequality is obtained by applying the following bound on binomial coefficients: $\binom{n}{d}\ge\left(\frac{n}{d}\right)^d$.
\end{proof}

\subsection{Simplex Tree}
  Let $K\in {\cal K}(n,d,k,m)$ be a simplicial complex whose  vertices are  labeled from 1 to $n$ and ordered accordingly. We can thus associate to each simplex of $K$ a word on the alphabet set $[n]$. Specifically, a $j$-simplex of $K$ is uniquely represented as the word of length $j + 1$ consisting of the ordered set of the labels of its $j + 1$ vertices. Formally, let $\sigma = \{v_{\ell_0} , \ldots , v_{\ell_j} \}$ be a simplex, where $v_{\ell_i}$ are vertices of $K$ and $\ell_i \in [n]$ and $\ell_0 <\cdot\cdot\cdot < \ell_j$ . $\sigma$ is represented by the word $[\sigma] = [ \ell_0 , \cdots , \ell_j ]$. The simplicial complex $K$ can be defined as a collection of words on an alphabet of size $n$. To compactly represent the set of simplices of $K$, the corresponding words are stored in a tree and this data structure is called the Simplex Tree of $K$ and denoted by ST$(K)$ or simply ST when there is no ambiguity. It may be seen as a trie on the words representing the simplices of the complex. The depth of the root is 0 and the depth of a node is equal to the dimension of the simplex it represents plus one. 

We give a constructive definition of ST. Starting from an empty tree,
insert the words representing the simplices of the complex in the
following manner. When inserting the word $[\sigma] = [ \ell_0
,\cdot\cdot\cdot, \ell_j ]$ start from the root, and follow the path
containing successively all labels $\ell_0 , \cdot\cdot\cdot ,
\ell_i$, where $[ \ell_0 ,\cdot\cdot\cdot, \ell_i ]$ denotes the
longest prefix of $[\sigma]$ already stored in the ST. Next,
append to the node representing $[ \ell_0 ,\cdot\cdot\cdot, \ell_i ]$
a path consisting of the nodes storing labels $\ell_{i+1}
,\cdot\cdot\cdot, \ell_j$. The filtration value of $\sigma$ denoted by $f(\sigma)$ is stored inside the node containing the label $\ell_j$, in the above path.
In Figure 2, we give ST for the simplicial complex shown in Figure~\ref{fig:SimplicialComplexExample}.

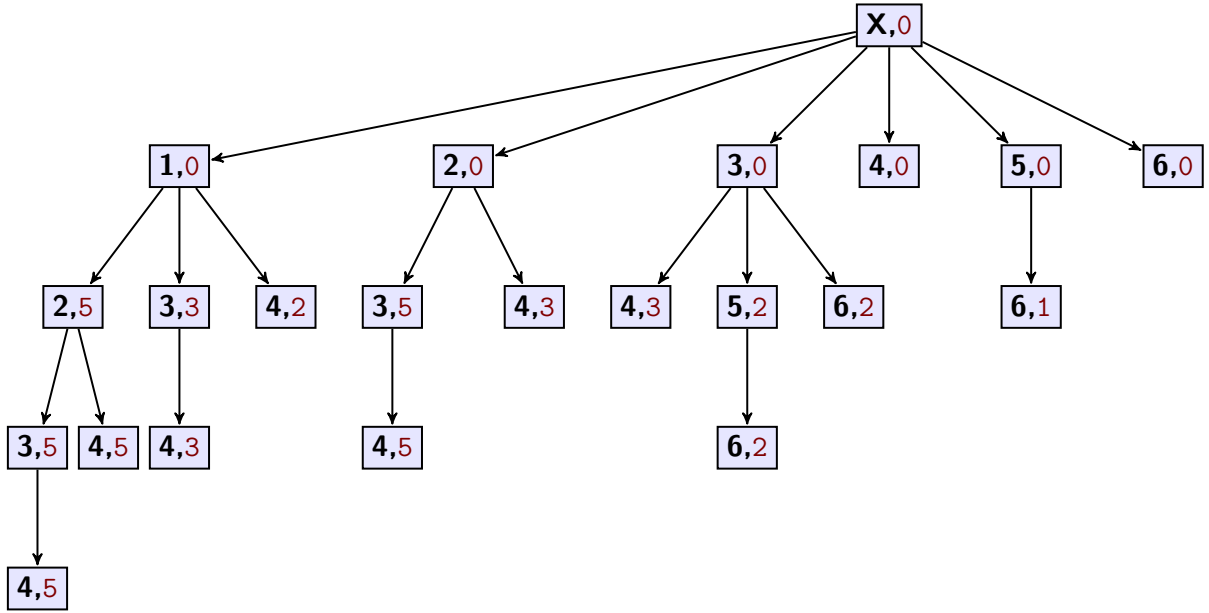
\begin{figure}[!h]
\centering
\resizebox{\textwidth}{!}{
\begin{tikzpicture}[->,>=stealth',shorten >=0.5pt,auto,node distance=4cm,
  thick,main node/.style={rectangle,fill=blue!10,draw,font=\sffamily\large\bfseries}]
  \node[main node] (0) at (0,0) {X,\color{red!50!black}\texttt{0}};
  \node[main node] (1) at (-10,-2) {1,\color{red!50!black}\texttt{0}};
  \node[main node] (2) at (-6,-2) {2,\color{red!50!black}\texttt{0}};
  \node[main node] (3) at (-2,-2) {3,\color{red!50!black}\texttt{0}};
  \node[main node] (4) at (0,-2) {4,\color{red!50!black}\texttt{0}};
  \node[main node] (5) at (2,-2) {5,\color{red!50!black}\texttt{0}};
  \node[main node] (6) at (4,-2) {6,\color{red!50!black}\texttt{0}};
  \node[main node] (12) at (-11.5,-4) {2,\color{red!50!black}\texttt{5}};
  \node[main node] (13) at (-10,-4) {3,\color{red!50!black}\texttt{3}};
  \node[main node] (14) at (-8.5,-4) {4,\color{red!50!black}\texttt{2}};
  \node[main node] (23) at (-7,-4) {3,\color{red!50!black}\texttt{5}};
  \node[main node] (24) at (-5,-4) {4,\color{red!50!black}\texttt{3}};
  \node[main node] (34) at (-3.5,-4) {4,\color{red!50!black}\texttt{3}};
  \node[main node] (35) at (-2,-4) {5,\color{red!50!black}\texttt{2}};
  \node[main node] (36) at (-0.5,-4) {6,\color{red!50!black}\texttt{2}};
  \node[main node] (56) at (2,-4) {6,\color{red!50!black}\texttt{1}};  
  \node[main node] (123) at (-12,-6) {3,\color{red!50!black}\texttt{5}};
  \node[main node] (124) at (-11,-6) {4,\color{red!50!black}\texttt{5}};
  \node[main node] (134) at (-10,-6) {4,\color{red!50!black}\texttt{3}};
  \node[main node] (234) at (-7,-6) {4,\color{red!50!black}\texttt{5}};
  \node[main node] (356) at (-2,-6) {6,\color{red!50!black}\texttt{2}};
  \node[main node] (1234) at (-12,-8) {4,\color{red!50!black}\texttt{5}};
      
  \path[every node/.style={font=\sffamily\small}]
    (0) edge node {} (1)
    (0) edge node {} (2)
    (0) edge node {} (3)
    (0) edge node {} (4)
    (0) edge node {} (5)
    (0) edge node {} (6)
    (1) edge node {} (12)
    (1) edge node {} (13)
    (1) edge node {} (14)
    (2) edge node {} (23)
    (2) edge node {} (24)
    (3) edge node {} (34)
    (3) edge node {} (35)
    (3) edge node {} (36)
    (12) edge node {} (123)
    (12) edge node {} (124)
    (13) edge node {} (134)
    (23) edge node {} (234)
    (35) edge node {} (356)
    (5) edge node {} (56)
    (123) edge node {} (1234)
     ;

\end{tikzpicture}
}
\caption{Simplex Tree of the simplicial complex in Figure~\ref{fig:SimplicialComplexExample}. In each node, the label of the vertex is indicated in black font and the filtration value stored by the simplex is in brown font.}
\label{fig:SimplexTreeExample}
\end{figure}

If $K$ consists of $m$ simplices (including the
empty face) then, the associated ST contains exactly $m$ nodes. Thus, we
need $\Theta(m(\log n+\log t))$ bits to represent the nodes in ST (since each node stores a vertex which needs $\Theta(\log n)$ bits to be represented and also stores the filtration value of the simplex that the node corresponds to, which needs $\Theta(\log t)$ bits to be represented). 
We can compare the upper bound obtained
to the lower bound of Lemma~\ref{FiltrationLowerbound}. In particular, ST matches the lower bound when $t$ is {$n^{\Omega (1)}$}. 

Now, we will briefly discuss the cost of doing some basic operations through ST on a simplicial complex. Checking if a simplex $\sigma$ is in the complex, is equivalent to checking the existence of the corresponding path starting from the root in ST. This can be done very efficiently in time $\mathcal{O}(d_\sigma \log n)$ and therefore all operations which primarily depend on the membership query can be efficiently performed using ST. One such example, is querying the filtration value of a simplex. However, due to its explicit representation,  insertion is a costly operation on ST (exponential in the dimension of the simplex to be inserted {since a $d$-simplex has $2^d$ faces}). Similarly, removal is also a costly operation on ST, since there is no efficient way to locate and remove all cofaces of a simplex. Consequently, topology preserving operations such as elementary collapse and edge contraction are also expensive for ST.  
These operation costs are summarized later in Table~\ref{tab:OperationsonMSD}.
In the next section, we will introduce a new data structure which does a better job of balancing between static queries (e.g. membership) and dynamic queries (e.g. insertion and removal)

\subsection{Motivation for a New Data Structure: {The Case of Flag Complexes}}
\label{motivation}
In this subsection, we will motivate the need for a new data structure 
that is more sensitive to the number of critical simplices than the Simplex Tree by taking up the case of flag complexes.  The flag complex of an undirected graph $G$ is defined as an abstract simplicial complex, whose simplices are the sets of vertices in the cliques of $G$.
As noted by Boissonnat et al.\ (see Section~2 in \cite{BKT15}), in several cases of interest, the number $k$ of maximal simplices of a flag complex is small.
In particular, 
we have $k=\mathcal{O}(n)$ for flag complexes constructed from planar graphs and expanders \cite{ELS10} and, more generally, from nowhere dense graphs \cite{GKS13}, and also from chordal graphs\cite{G80}. Furthermore, 
$k=n^{\mathcal{O}(1)}$ \cite{GKS13} for all flag complexes constructed from graphs with degeneracy $\mathcal{O}(\log n)$ (degeneracy is the smallest integer $r$ such that every subgraph has a vertex of degree at most $r$).

We further add to this list of observations by noting that the flag complexes of $K_{\ell}$-free graphs have at most $\max\{n,n\Delta^{\ell-2}/2^{\ell-2}\}$ maximal simplices \cite{P95}, where $\Delta$ is the maximum degree of any vertex in the graph. Thus, when $\Delta$ and $\ell$ are constants, we have $k=\mathcal{O}(n)$.  Finally, we note that the flag complexes of Helly circular-arc (respectively, circle) graphs \cite{G74,D03}, and boxicity-2 graphs \cite{S03} have $k=n^{\mathcal{O}(1)}$ from Corollary~4 of \cite{RS07}.
This encompasses a large class of complexes
encountered in practice and if the number of maximal simplices is small, then any data structure that is sensitive to $k$ is likely to be more efficient than \ST. Moreover, we go a step ahead and show below that any data structure that is sensitive to $\kappa$, the number of critical simplices, is likely to be more efficient than \ST\ as well. This is done by proving that in most of the above cases, if $k$ is small then so is $\kappa$. We formalize the argument below.

{We recall the definition of a hereditary property of a graph \cite{H73}. A property $P$ of a graph $G$ is hereditary if, whenever $G$ has property $P$ and $G'$ is a subgraph of $G$, then $G'$ also has property $P$.  We consider a specific filtration $f$ for flag complexes which is of significant interest: the filtration value of every edge $e$ in the complex ({denoted $f(e)$})  is given as part of the input and the filtration value of a simplex of higher dimension is equal to the maximum of the filtration values of all the edges in the simplex. We have the following lemma  for the filtration described above.}

\begin{lemma}\label{flagpsi}
Let $t:\mathbb N\to\mathbb N$ be a non-decreasing function. Let $P$ be a hereditary property of a graph such that any graph $G$ on $n$ vertices and $E$ edges with the property $P$ has at most $t(n)$ maximal cliques. Then the number of critical simplices of the flag complex of $G$ is at most $n + t(n) \cdot  E$.
\end{lemma}
\begin{proof}
For every edge $e$ in $G$, let $G_e$ be the subgraph of $G$ induced by the vertices of $e$ and all the vertices which are common neighbors of both the vertices of $e$. Then, we remove all the edges in $G_e$ which have a filtration value higher than $f(e)$. Let $M_e$ be the set of all maximal cliques in $G_e$. Let $M=\underset{e}{\bigcup}\ M_e$. 
Let $\sigma$ be a critical simplex of filtration value $f(e)$. We claim that $\sigma$ is a maximal clique in $G_e$. To see that $\sigma$ is a clique in $G_e$, notice that from the monotonicity property of filtration, we have that the filtration value of the edges in $\sigma$ are at most $f(e)$, and thus all the edges of $\sigma$ are in $G_e$. The observation that  $\sigma$ is a maximal clique in $G_e$  follows from the assumption that $\sigma$ is a critical simplex in the flag complex.

Hence, we have that the number of critical simplices is at most $|M|+n$ (the additive $n$ term comes from the possibility that the vertices might be critical simplices). Thus, we have the following upper bound on the number of critical simplices. 
\begin{align}
|M|&\le E\cdot \max_{e}\ |M_e|\nonumber\\
&\le E\cdot \max_{e}\ t\left(|G_e|\right)\label{eqher1}\\
&\le E\cdot t(n),\label{eqher2}
\end{align}
where \eqref{eqher1} follows from the hereditary property of $G$ and \eqref{eqher2} follows from $t$ being a non-decreasing function.
\end{proof}

We note that planarity, expansion (in some cases) and degeneracy are all hereditary properties of a graph. 
Thus from the above lemma we have that a data structure sensitive to the number of critical simplices representing the  flag complexes of graphs having any of the above properties is not only of small size but may also be efficient to construct (i.e., requires $\text{poly}(n)$ time).
{Our goal is to show that the Critical Simplex Diagram to be introduced in the next section  exhibits these nice properties for a large class of general simplicial complexes.}

\section{Critical Simplex Diagram}\label{CSD}
In this section, we introduce the {\em Critical Simplex Diagram}, $\CSD(K)$, which is an improved version of SAL \cite{BKT15} (as it has the additional functionality of maintaining filtrations). CSD is a  collection of $n$ arrays that correspond to the $n$ vertices of $K$. The elements of an array, referred to as {\em nodes} in the rest of the paper, correspond to copies of the vertex of $K$ associated to the array.  The collection of all these nodes across all the arrays will represent the disjoint union of the critical simplices in the complex.
Additionally, CSD has  edges that join nodes of different arrays. Each connected component of edges in CSD represents a (critical) simplex of $K$.   We describe some notations used throughout the section below. 

\subsection{Notations}\label{sec:notations}

Through out the section, we will denote by  $\sigma$  a simplex of dimension $d_{\sigma}$ in the complex and  denote its vertices by  $[v_{\ell_0} \cdots v_{\ell_{d_\sigma}}]$.
Let $S_h$ be the set of simplices in the complex whose filtration value is $h$. Let $M_h$ be the subset of $S_h$ containing \emph{all} the critical simplices of the complex in $S_h$. For instance, in the complex of Figure~\ref{fig:SimplicialComplexExample}, we have $M_0=\{1,2,3,4,5,6\}$, $M_1=\{[56]\}$, $M_2=\{[14],[356]\}$, $M_3=\{[134], [24]\}$, $M_4=\emptyset$, and $M_5=\{[1234]\}$. Moreover, we note that $M_0, M_1,...., M_t$ are all disjoint and denote by $M$ the union of all $M_h$. Hence $M$ is the set of all critical simplices of the complex.

We denote by $\Psi_{\max}(i)$ the subset of simplices of $M$ that contain  vertex $i$. In other words,
$$\Psi_{\max}(i)={\{\sigma\in K\big| i\in \sigma, f(\sigma)<f(\tau), \forall\tau\in K, \;\; \sigma \subset \tau\}.}	
$$
Let $\Psi=\underset{i\in [n] }{\max}\ |\Psi_{\max}(i)|$.

{We denote by  $\Gamma_j$ the  largest number of maximal simplices of $K$ that contains any $j$-simplex of $K$.} We have the following bounds: $$k\ge \Gamma_0\ge\Gamma_1\ge \cdots \ge \Gamma_d = 1,$$ $$\Gamma_0\le \Psi \le m.$$
Moreover, when $t=0$, we have $\Psi=\Gamma_0$. 
We describe the construction of CSD below.

\subsection{Construction}\label{sec:construction}
We initially have $n$ empty arrays $A_1,\ldots ,A_n$. 
The vertices of $K$ are associated to the arrays. Each array contains a set of nodes that are copies of the vertex associated to the array and are labelled by an ordered pair of integers (to be defined below).
Nodes belonging to distinct arrays are joined by edges leading to a graph structure. Each connected component of that graph is a star tree\footnote{A star tree is a tree in which there is a special vertex with an edge to every other vertex in the graph.} and there is a bijection between the star trees of that graph and the set of critical simplices {$M$}. All the nodes of such a simplex are
labelled by the same ordered pair of integers.  The first integer refers to the filtration value of the simplex and the second integer refers to a number used to index critical simplices that have the same filtration value. 
For instance, in Figure~\ref{fig:CSDExample} we have the CSD representation of the simplicial complex of Figure~\ref{fig:SimplicialComplexExample}, and the triangle $[134]$ in $M_3$ is represented by 3 nodes, each with label $(3,1)$, that are  connected by edges. 
Below we provide a more detailed treatment of the construction of CSD.

 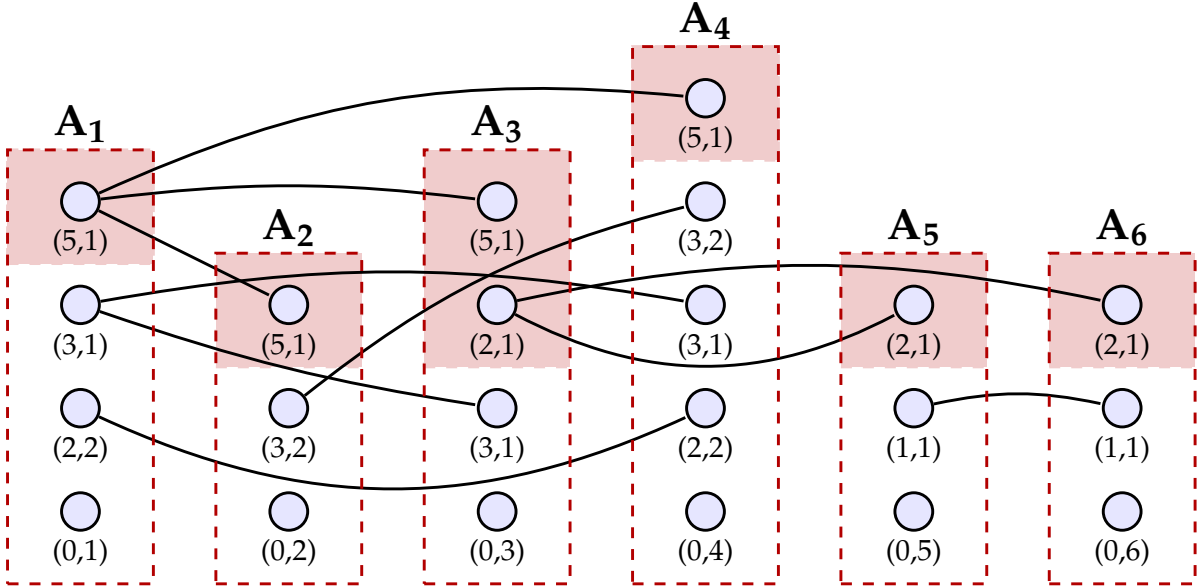
\begin{figure}[!h]\centering
\resizebox{\textwidth}{!}{
 
\begin{tikzpicture}[->,>=stealth',shorten >=0.5pt,auto,node distance=2cm,
  thick,main node/.style={circle,fill=blue!10,draw,font=\sffamily\large\bfseries}]

     \filldraw[red!70!black!20!white,dashed] (3.7,2.5) rectangle (2.3,1.4);
     \filldraw[red!70!black!20!white,dashed] (1.7,3.5) rectangle (0.3,2.4);     
     \filldraw[red!70!black!20!white,dashed] (4.3,1.4) rectangle (5.7,3.5);
     \filldraw[red!70!black!20!white,dashed] (6.3,3.4) rectangle (7.7,4.5);     
     \filldraw[red!70!black!20!white,dashed] (8.3,1.4) rectangle (9.7,2.5);
     \filldraw[red!70!black!20!white,dashed] (10.3,1.4) rectangle (11.7,2.5);     

  \node[main node] (122) at (1,1) {};
  \node[main node] (1) at (1,0) {};
  \node[main node] (2) at (3,0) {};
  \node[main node] (3) at (5,0) {};    
  \node[main node] (4) at (7,0) {};
  \node[main node] (5) at (9,0) {};
  \node[main node] (6) at (11,0) {};      
  \node[main node] (131) at (1,2) {};  
  \node[main node] (151) at (1,3) {};
  \node[main node] (232) at (3,1) {};  
  \node[main node] (251) at (3,2) {};
  \node[main node] (331) at (5,1) {};
  \node[main node] (321) at (5,2) {};  
  \node[main node] (351) at (5,3) {};  
  \node[main node] (422) at (7,1) {};
  \node[main node] (431) at (7,2) {};  
  \node[main node] (432) at (7,3) {};
  \node[main node] (451) at (7,4) {};
  \node[main node] (511) at (9,1) {};
  \node[main node] (521) at (9,2) {};  
  \node[main node] (611) at (11,1) {};
  \node[main node] (621) at (11,2) {};

  \path[every node/.style={font=\sffamily}]
    (122) edge  [-,bend right=25
    ] (422)

    (151) edge [-
    ] (251)
    (151) edge [-,bend left=7
    ] (351)
    (151) edge [-,bend left=15
    ] (451)

    (321) edge [-,bend right=28
    ] (521)
    (321) edge [-,bend left=12
    ] (621)

    (131) edge [-,bend right=5
    ] (331)
    (131) edge [-,bend left=10
    ] (431)

    (232) edge [-,bend left=12
    ] (432)

    (511) edge [-,bend left=12
    ] (611)

	

    ;
     
	\scriptsize	

     \draw[red!70!black,dashed] (3.7,2.5) rectangle (2.3,-0.7);
     \draw[red!70!black,dashed] (1.7,3.5) rectangle (0.3,-0.7);     
     \draw[red!70!black,dashed] (4.3,-0.7) rectangle (5.7,3.5);
     \draw[red!70!black,dashed] (6.3,-0.7) rectangle (7.7,4.5);     
     \draw[red!70!black,dashed] (8.3,-0.7) rectangle (9.7,2.5);
     \draw[red!70!black,dashed] (10.3,-0.7) rectangle (11.7,2.5);     
	 \node at (1,-0.4) {(0,1)};
	 \node at (3,-0.4) {(0,2)};
	 \node at (5,-0.4) {(0,3)};
	 \node at (7,-0.4) {(0,4)};
	 \node at (9,-0.4) {(0,5)};
	 \node at (11,-0.4) {(0,6)};

	 \node at (1,0.6) {(2,2)};
     \node at (1,1.6) {(3,1)};          
     \node at (1,2.6) {(5,1)};
	 \node at (3,0.6) {(3,2)};
     \node at (3,1.6) {(5,1)};          
     \node at (5,0.6) {(3,1)};          
	 \node at (5,1.6) {(2,1)};
     \node at (5,2.6) {(5,1)};
     \node at (7,0.6) {(2,2)};
     \node at (7,1.6) {(3,1)};
     \node at (7,2.6) {(3,2)};
     \node at (7,3.6) {(5,1)};     
     \node at (9,0.6) {(1,1)};
     \node at (9,1.6) {(2,1)};     
     \node at (11,0.6) {(1,1)};
     \node at (11,1.6) {(2,1)};

	\large{\node at (1, 3.75) {$\mathbf{A_1}$};     }
	\large{\node at (3, 2.75) {$\mathbf{A_2}$};     }
	\large{\node at (5, 3.75) {$\mathbf{A_3}$};     }
	\large{\node at (7, 4.75) {$\mathbf{A_4}$};     }
	\large{\node at (9, 2.75) {$\mathbf{A_5}$};     }
	\large{\node at (11, 2.75) {$\mathbf{A_6}$};     }		
				
\end{tikzpicture}
	}
\caption{Critical Simplex Diagram for the complex in Figure~\ref{fig:SimplicialComplexExample}. The shaded region in each $A_i$ corresponds to $A_i^\star$. { A star tree represents a critical simplex. Each node has two labels, the first one is the filtration value and the second one designates a unique identification to the simplex within the set of critical simplices with the same filtration value.}}

  \label{fig:CSDExample}

\end{figure}

Given $M_h$ for every $h\in\llbracket t\rrbracket$, we build the CSD by inserting the simplices in $M_h$ in decreasing values of $h$. For {every} simplex $\sigma=v_{\ell_0} \cdot\cdot\cdot v_{\ell_{d_{\sigma}}}$ in $M_h$, we associate a unique key ${\cal H}(\sigma) \in \left[|M_h|\right]$. We insert a node with label $(h,\mathcal{H}(\sigma))$ into each of the arrays $A_{\ell_0},\ldots ,A_{\ell_{d_{\sigma}}}$. 
For every $j\in [d_{\sigma}]$, we introduce an edge between $(h,\mathcal{H}(\sigma))$ in $A_{\ell_{0}}$ and $(h,\mathcal{H}(\sigma))$ in $A_{\ell_{j}}$. In other words, a critical simplex  $\sigma$ is represented in CSD by a star tree in the graph thus defined. Each star tree is on $d_{\sigma}+1$ nodes where $v_{\ell_0}$ is the center of the star tree. Furthermore, each node in CSD corresponds to a vertex in exactly one simplex.

We denote by $A_i^\star$ the set of all nodes in $A_i$ which are part of the star trees that correspond to maximal simplices in $K$ (the region with these nodes are shaded in Figure~\ref{fig:CSDExample}). Inside each $A_i$, we first sort nodes based on whether they are in $A_i^\star$ or not. Furthermore, inside $A_i^\star$ and inside $A_i\setminus A_i^\star$, we sort the nodes according to the lexicographic order of their labels.  
Thus we have that $A_i^\star$ is a contiguous subarray of $A_i$, i.e., all consecutive elements in $A_i^\star$ are also consecutive elements in $A_i$, as 
can be observed in Figure~\ref{fig:CSDExample}. We partition $A_i$ as above because it will help us perform the operations which do not depend on the filtration of the complex (such as membership query) at least as fast as SAL.

Finally, we remark here that we can use a priority queue data structure to efficiently generate keys for the simplices.

\subsection{Size of the Critical Simplex Diagram}\label{sizeofCSD}
We denote the total number of nodes in CSD by $|\CSD|$.
The number of nodes in each $A_i$ is at most $\Psi$,  and thus $|\CSD|$ is at most $\Psi n$.  Note that the number of edges in CSD is also at most $\Psi n$ since CSD is a collection of star trees (i.e., a star forest).

 Alternatively, we can bound the number of nodes by $|M|d$, where $M$ as defined in Section~\ref{sec:notations} is the union of all $M_h$. The actual relation between $\Psi$ and $|M|$ can be stated as follows:
$$\sum_{i=1}^{n} |\Psi_{\max}(i)| = \sum_{\sigma\in M} (d_\sigma+1).$$
{We recall that in the previous  section we denoted $|M|$ by $\kappa$.}

Furthermore, in each node we store a filtration value (which requires $\log t$ bits), a unique key (which requires $\underset{h\in\llbracket t\rrbracket}{\max}\ \log |M_h|\le \log m$ bits), and a pointer to the center vertex of the star tree. If the node corresponds to the center vertex of a star tree then it additionally maintains the list of all its neighbors as well. We can thus upper bound the space needed to store the nodes of CSD by $\kappa d(3\log m+\log t)$.
 Thus, CSD matches the lower bound in Theorem~\ref{CSDLowerBound}, up to a multiplicative factor of $\mathcal O(d)$ (as $m=\mathcal{O}\left(n^d\right)$).

In the case of $\CSD$, we are interested in the value of $\Gamma_0$ and $\Psi$ which we use to estimate the worst-case cost of basic operations in $\CSD$, in the following subsection.

\subsection{Operations on the Critical Simplex Diagram}\label{sec:operations}
Let us now analyze the cost of performing basic operations on $\CSD$. 
First, we describe how to intersect arrays and update arrays in CSD as these are elementary operations on CSD which will be required to perform basic operations on the simplicial complex that it represents. Next, we describe how to perform static queries such as the membership query. Then we describe how to perform dynamic queries such as the insertion or  removal of a simplex. Finally, we compare the efficiency of CSD with ST. We remark here that in order to perform the above operations efficiently, we will exploit the fact that the filtration value of a simplex that is not critical is equal to the minimum of the filtration values of its cofaces. 

\subsubsection{Elementary Operations on the Critical Simplex Diagram}

\subsubsection*{Implementation of the arrays}
We implement the arrays $A_i$ using  red-black trees as the building blocks. This allows us to efficiently search, insert, and remove an element inside $A_i$ in time $\mathcal{O}(\log |A_i|)$.  {We treat 
 $A_i^\star$ and  $A_i\setminus A_i^\star$ separately. 
Below we will discuss how to implement $A_i^\star$ and the same will hold for $A_i\setminus A_i^\star$. Each subarray of $A_i^\star$ with a same filtration value, i.e., the same first coordinate, is implemented using a red-black tree. These subarrays which partition $A_i^\star$ are labelled with the common first coordinate value of their elements. We represent the set of subarrays using a red-black tree on their labels. Therefore, each $A_i^\star$ is represented as the composition of two  
red-black trees.  $A_i\setminus A_i^\star$ is represented similarly.} Finally, to search in $A_i$  we sequentially search in $A_i^\star$ and then search in $A_i\setminus A_i^\star$.

\subsubsection*{Intersecting arrays}

We will need to compute $A_\sigma$, defined as the intersection of $A_{\ell_0},\dots ,A_{\ell_{d_\sigma}}$, and  $A_\sigma^\star$, defined as the intersection of $A_{\ell_0}^\star,\dots ,A_{\ell_{d_\sigma}}^{\star}$. To compute  $A_\sigma$, we first find out the array with fewest elements amongst $A_{\ell_0},\dots ,A_{\ell_{d_\sigma}}$. Then, for each element $x$  in that array, we search for $x$ in the other $d_\sigma$ arrays, which can be done in time 
$\mathcal{O}\left(d_\sigma\log \left(\underset{i}{\max}\ |A_{\ell_i}|\right)\right)$. Hence $A_\sigma$ can be computed in time 
$$\mathcal{O}\left(\left(\underset{i}{\min}\ |A_{\ell_i}|\right)d_\sigma\log \left(\underset{i}{\max}\ |A_{\ell_i}|\right)\right) = O(d_{\sigma} \Psi\log\Psi).$$
We can compute 
$A_\sigma^\star$ in the same way as  $A_\sigma$ in time $$\mathcal{O}\left(\left(\underset{i}{\min}\ |A_{\ell_i}^\star|\right)d_\sigma\log \left(\underset{i}{\max}\ |A_{\ell_i}^\star|\right)\right)
 = O(d_{\sigma}\Gamma_0\log {\Gamma_0}).$$

\subsubsection{Static Operations on the Critical Simplex Diagram}\label{staticoperations}

The tree structure of ST provides an efficient representation to perform static operations. However, we show below that we are able to answer static queries using CSD by only paying a multiplicative factor of $\Psi$ (in the worst case) over the cost of performing the same operation in ST. In the case of the membership query, the multiplicative factor is reduced to $\Gamma_0$.

\subsubsection*{Membership of a Simplex} 
We first observe that  $\sigma\in K$ if and only if $A_\sigma^\star\neq\emptyset$. This is because if $\sigma\in K$, then there exists a maximal simplex in $K$ that contains $\sigma$. The star tree associated to this maximal simplex  has nodes in all the $A_{\ell_i}^\star$, and all those nodes have the same label. This implies that $A_\sigma^\star\neq\emptyset$, and the converse is also true.
It follows that  determining if $\sigma$ is in $K$ reduces to computing $A_\sigma^\star$ and checking whether it is non-empty. This procedure is very similar to the analogous procedure using SAL \cite{BKT15}. Therefore, membership of a simplex can be determined in time $\mathcal{O}(d_\sigma\Gamma_0\log {\Gamma_0})$. Finally, we note that the membership query allows us to decide if a simplex is maximal in the complex since {a simplex is maximal iff $|A_{\sigma}^{*}|=1$ and the star tree corresponding to the label in $A_{\sigma}^{*}$ has exactly $d_\sigma+1$ nodes. }  We denote this new query by \texttt{is\_maximal}. 

\subsubsection*{Access Filtration Value} 
Given a simplex  $\sigma$  of $K$ we want to access its filtration value $f(\sigma)$.  
We observe that $f(\sigma)$ is the minimal filtration value of the nodes in $A_\sigma$ since the  filtration function is monotone w.r.t.\ inclusion. Hence, accessing the filtration value of $\sigma$ reduces to computing $A_\sigma$. Therefore, the filtration value of a simplex can be determined in time $\mathcal{O}(d_\sigma\Psi\log \Psi)$.

For example, consider the $\CSD$ in Figure~\ref{fig:CSDExample}. We have to find the filtration value of $\sigma=[134]$ in the complex of Figure~\ref{fig:SimplicialComplexExample}.  We see that $A_1\cap A_3 \cap A_4=\{(3,1),(5,1)\}$. This means that the filtration value of the triangle is $f([134])=\min\ (3,5)=3$. Finally, we note that the filtration query allows to decide if a simplex is critical in the complex since  {a simplex is critical iff $|A_{\sigma}|=1$ and the star tree corresponding to the label in $A_{\sigma}$ has exactly $d_\sigma+1$ nodes. } This new query is denoted by \texttt{is\_critical}.  

\subsubsection*{Computing Filtration Value of Facets} {Given a simplex $\sigma= [v_{l_0}\cdots v_{l_{d_{\sigma}}}]$, we can provide the filtration values of all of its $d_\sigma+1$ facets by calling the previous operation for each of its facets, which will require a total running time of $\mathcal{O}(d_\sigma^2 \Psi\log \Psi)$. We can improve on this naive approach
and obtain the filtration values of the $d_\sigma+1$ facets of $\sigma$ in running time of $\mathcal{O}(d_\sigma \Psi\log \Psi)$. At a high level, we find an efficient way to compute the filtration values of all but one facet in time essentially equal to the cost of performing one access filtration value query. We compute the filtration value of the remaining facet using the access filtration value query.

For $j\in \llbracket d_\sigma\rrbracket$, we denote by $\sigma_{\ell_j}$ the facet of $\sigma$ opposite to $v_{\ell_j}$, i.e. $\sigma_{\ell_j}=v_{\ell_0},\ldots ,v_{\ell_{j-1}},v_{\ell_{j+1}},\ldots ,v_{\ell_{d_\sigma}}$.
Let $r\in\llbracket d_\sigma\rrbracket$ be such that $r=  \arg\min_i    
 \lvert A_{\ell_i}\rvert$. 
Let $F$ be the set of facets of $\sigma$ that are incident to $v_{\ell_r}$ and $B$  the set of critical simplices that contain a facet of $F$ but do not contain $\sigma$.
Equivalently, $B$ is the subset of $A_{\ell_r}$ containing every element of $A_{\ell_r}$ which appears in exactly $d_{\sigma}-1$ of the sets in $A_{\ell_0},\ldots ,A_{\ell_{r-1}},A_{\ell_{r+1}},\ldots ,A_{\ell_{d_\sigma}}$.

Hence  the number of elements of $B$ is $O(\Psi)$ and $B$ can be computed in time $\mathcal{O}(|A_{\ell_r}|d_\sigma\log\Psi ) = \mathcal{O}(\Psi d_{\sigma}\log\Psi )$. 

We define a map $g$ that  maps every simplex in $B$ to the facet of $\sigma\in F$ that it contains. 
We sort $B$ based on $g$, and in case of ties based on the first coordinate. For every $j\in \llbracket d_\sigma\rrbracket$, $j\neq r$, let $\alpha_j$ be the minimal filtration value of the nodes in $B$ which are mapped to $j$ under $g$. Then, the filtration value of the facet $\sigma_{\ell_j}$ 
is the minimum between $\alpha_j$ and $f(\sigma)$. We can compute $f(\sigma)$ using the function Access Filtration Value described previously. If there are no nodes in $B$ mapped to $j$ under $g$ then the filtration value of $\sigma_{\ell_j}$ is $f(\sigma)$.  From the above discussion, we see that computing the filtration values of the $\sigma_{\ell_j}$ for all $j\neq r$, i.e.  for the $d_\sigma$ facets of $\sigma$ incident to vertex $v_{\ell_r}$, requires a total time of $\mathcal{O}(|B|\log |B|+d_\sigma\log |B|)=\mathcal{O}((\Psi+d_\sigma)\log \Psi)$. 
It remains to compute the filtration value of $\sigma_{\ell_r}$ which can be done using the function Access Filtration Value.
We conclude that the total running time to compute the filtration values of all the facets of $\sigma$  is $\mathcal{O}(d_\sigma \Psi\log \Psi)$.}

\subsubsection*{Computing Filtration Value of Cofaces of codimension 1} 

Given a simplex $\sigma$, we perform the access filtration value query to obtain the list of all the critical simplices that contain $\sigma$ (there are at most $\Psi$ such simplices).
For each critical simplex $\tau$ that contains $\sigma$, we list all its faces which are cofaces of $\sigma$ of codimension 1 (there are at most $d-d_{\sigma}$ of them), and associate to each of them a filtration value of $f(\tau)$. 
This can be done in time $\mathcal O(\Psi d\log \Psi)$ by traversing through the star tree corresponding to $\tau$.
It might happen that a coface of $\sigma$ of codimension 1 can be associated to more than one filtration value (because this coface is contained in more than one critical simplex). In such a case, the filtration value of the coface is associated to the minimum of the filtration values of the critical simplices that contain it. Again, this can be done in $\mathcal O(\Psi d\log \Psi)$ time. Therefore, the filtration value of all the cofaces of a simplex of codimension 1 can be computed in $\mathcal O(\Psi d\log \Psi)$ time.

\subsubsection{Dynamic Operations on the Critical Simplex Diagram}
\label{sec:dynamicoperations}
Now, we will see how to perform dynamic operations on CSD. We note here that CSD is more suited to perform dynamic queries compared to ST because of its non-explicit representation, and this means that the amount of information to be modified in CSD is always less than in ST. 

\subsubsection*{Lazy Insertion}  {When inserting a critical simplex $\sigma$, we need to check if some of its faces 
were critical (prior to the insertion of $\sigma$)  and have the same filtration value as $f(\sigma)$. We need to remove such simplices since they are no longer maximal for this value of the filtration after the insertion of $\sigma$. The full operation will be described in the next paragraph.  We first describe a lazy version that do not remove (nor even consider) the subfaces of $\sigma$.  We further assume that we know whether $\sigma$ is maximal or not. Lazy insertions will be extensively used in the later sections. 
A lazy insertion in CSD requires $\mathcal{O}(d_\sigma\log\Psi)$ time (i.e., the cost of updating the arrays). We remark here that lazy insertions will not hinder any subsequent operation on CSD nor modify their time complexity, and, in order to save memory space,  we can  clean-up the data structure.  
See details in the paragraph called `Robustness in Modification' in Section~\ref{Performance}.}

\subsubsection*{Insertion}  Let, as before, $\sigma$ be a simplex that we want to insert in CSD. {Since CSD only stores critical simplices, we can assume that the filtration value $f(\sigma)$ of $\sigma$ is less than the filtration values of all its cofaces.} The insertion operation consists of first checking if $\sigma$ is a maximal simplex in $K$ using \texttt{is\_maximal}. 
 
\begin{sloppypar}{\em If $\sigma$ is a maximal simplex}, we insert the star tree corresponding to $\sigma$ in $A_{\ell_0}^\star,\ldots ,A_{\ell_{d_\sigma}}^\star$. Updating the arrays $A_{\ell_i}$ takes time $\mathcal{O}(d_\sigma\log \Psi)$. Next, we have to check if there existed maximal simplices in $K$ which are now faces of $\sigma$ (and therefore no longer maximal). We  either remove such a simplex if its filtration value is equal to $f(\sigma)$ or move it from $A_{\ell_i}^\star$ to $A_{\ell_i} \setminus A_{\ell_i}^\star $ if its filtration value is strictly less than that of $f(\sigma)$. We restrict our search to the faces of $\sigma$ which were previously {maximal}, by considering for every vertex $v$ in $\sigma$, the set of all maximal simplices that contain $v$, denoted by $Z_v$. We can compute $Z_v$ in time $\mathcal{O}( d_{\sigma}\Gamma_0 \log \Gamma_0)$. Then, we compute ${\cup}_{v\in\sigma} Z_v$ whose size is at most $(d_\sigma+1)\Gamma_0$ and check if any of these maximal simplices is a face of $\sigma$ (which can be done in $\mathcal{O}(d_\sigma^2\Gamma_0)$ time). If such a face of $\sigma$ in ${\cup}_{v\in\sigma} Z_v$ has filtration value equal to $f(\sigma)$ then, we remove its associated star tree. To remove all such star trees takes time $\mathcal{O}(d_\sigma^2\Gamma_0\log \Gamma_0)$. On the other hand, if the filtration value of the face is less than $f(\sigma)$ then, we will have to move the node from $A_{\ell_i}^\star$ to $A_{\ell_i} \setminus A_{\ell_i}^\star $ and to place it appropriately to maintain the sorted structure of $A_{\ell_i}$.  To reallocate all such star trees takes time $\mathcal{O}(d_\sigma^2\Gamma_0\log \Psi)$. Summarizing, to  remove or reallocate the faces of $\sigma$ which were previously maximal takes time at most $\mathcal{O}(d_\sigma^2\Gamma_0(\log \Gamma_0+\log \Psi))$. Next, we have to remove all the faces of $\sigma$ which were critical and whose filtration values were equal to $f(\sigma)$. This can be done in time $\mathcal{O}(\Psi d_\sigma^2\log \Psi)$ by simply identifying the star trees all of whose nodes are contained in $A_{\ell_0}\cup\cdots \cup A_{\ell_{d_\sigma}}$.  \end{sloppypar}

{\em If $\sigma$ is not a maximal simplex} then, we insert the star tree corresponding to $\sigma$ in $A_{\ell_0},\ldots ,A_{\ell_{d_\sigma}}$. Updating the arrays $A_{\ell_i}$ takes time $\mathcal{O}(d_\sigma\log \Psi)$. Next, we have to remove all the faces of $\sigma$ which were critical and whose filtration value were equal to $f(\sigma)$. This can be done in time $\mathcal{O}(\Psi d_\sigma^2\log \Psi)$ as described in the previous case. Thefore, the total running time in this case is $\mathcal{O}(\Psi d_\sigma^2 \log \Psi)$. 

\sloppypar{We conclude that the total time for inserting a simplex $\sigma$ is $\mathcal{O}(d_\sigma^2\Gamma_0 (\log \Gamma_0 +\log \Psi)+\Psi d_\sigma^2 \log \Psi)=\mathcal{O}(\Psi d_\sigma^2\log \Psi)$.}

\subsubsection*{Removal} 
To remove a face $\sigma$, we first perform an access filtration value query of $\sigma$ (requires $\mathcal{O}(\Psi d_{\sigma}\log \Psi)$ time). Note that there are at most $\min_{i\in\llbracket d_\sigma\rrbracket}|A_{\ell_i}|$ many critical simplices that contain $\sigma$. We deal with the simplices in  $ A_\sigma^\star$ and $ A_\sigma\setminus A_\sigma^\star$ separately. 
 For every simplex $ \tau\in A_\sigma^\star$, i.e., for every coface $\tau$ of $\sigma$ in $K$ which is a maximal simplex, we remove its corresponding star tree from the CSD. Since there are at most $\Gamma_{d_{\sigma}}$ maximal simplices that contain $\sigma$, the above removal of star trees can be done in $\mathcal{O}(\Gamma_{d_{\sigma}}d\log \Psi)$ time. Next, for each maximal simplex $\tau$ (containing $\sigma$) that we removed, and for every $i\in \llbracket d_\sigma\rrbracket $ we check if the facet of $\tau$ obtained by removing  $v_{\ell_{i}}$ from $\tau$, is a maximal simplex. If yes, we lazy insert the facet as a maximal simplex with the same filtration value as $\tau$. If no,  we lazy insert the facet  as a non-maximal simplex with the same filtration value as $\tau$ provided that it is still a critical simplex (can be checked using the \texttt{is\_critical} query). Note that in order to check if the above mentioned $d_\sigma+1$ facets of $\tau$ are maximal/critical, we do not have to make $d_{\sigma}+1$ \texttt{is\_maximal}/\texttt{is\_critical} queries, but can do the same checking in $\mathcal{O}(\Gamma_{d_\sigma}d\log \Psi )$/$\mathcal{O}(\Psi d\log \Psi )$ time by using the same idea that is described in the `computing filtration value of facets' paragraph in Section~\ref{staticoperations}.

Next, for every simplex $ \tau\in A_\sigma\setminus A_\sigma^\star$  i.e., for every coface $\tau$ of $\sigma$ in $K$ which is a critical (not maximal) simplex, we replace its corresponding star tree by star trees of its $d_{\sigma}+1$ facets with the same filtration value, where the $i^{\text{th}}$ facet is obtained by removing  $v_{\ell_{i-1}}$ from $\tau$. Introducing a star tree and updating the arrays $A_{\ell_i}$ takes time $\mathcal{O}(d_\tau\log \Psi)$. Again, we note that we introduce the star tree of the facet of $\tau$ if and only if the facet is still a critical simplex in the complex (can be checked using the \texttt{is\_critical} query).
Furthermore, if $\sigma$ is a critical simplex (can be checked by \texttt{is\_critical} query) then, we know that there is a star tree representing $\sigma$. We replace this star tree by the star tree for each of its facets which have the same filtration value (if the facets are critical). 

Therefore, the total time for removal is $\mathcal{O}(\Psi^2 d d_\sigma\log \Psi)$.

\subsubsection*{Elementary Collapse} A simplex $\tau$ is collapsible through one of its facets $\sigma$,
if $\tau$ is the only coface of $\sigma$. Such
a pair $(\sigma,\tau)$ is called a free pair, and removing both faces of a free pair is an elementary
collapse. Given a pair of simplices $(\sigma,\tau)$, to check if it is a free pair is done by obtaining the list of all maximal simplices that contain $\sigma$, through the membership query (costs $\mathcal{O}(d_\sigma \Gamma_0\log\Gamma_0 )$ time) and then checking if $\tau$ is the only member in that list that is different from $\sigma$. If yes, then we remove $\tau$ from the CSD by just removing all the nodes in the corresponding arrays in time $\mathcal{O}(d_{\tau}\log \Psi)$. We call such an operation (of removal) as lazy removal. Next, for every facet $\sigma^\prime$ of $\tau$ other than $\sigma$, we check if $\sigma^\prime$ is a critical simplex (post the removal of $\tau$) by applying \texttt{is\_critical}. If yes, we lazy insert $\sigma^\prime$ in time $\mathcal{O}(d_\sigma\log \Psi)$. Finally, if $\sigma$ is explicitly represented in CSD (i.e., it was a critical simplex prior to the removal of $\tau$) then we lazy remove it and, for every facet of $\sigma$,  we similarly check if it is a critical simplex (post the removal of $\sigma$) by applying \texttt{is\_critical}. If yes, we lazy insert that facet in time $\mathcal{O}(d_\sigma\log \Psi)$.
Thus, the total running time is $\mathcal{O}(d_\sigma (d_\sigma\Psi\log \Psi+\Gamma_0\log\Gamma_0))=\mathcal{O}(d_\sigma^2 \Psi\log \Psi)$.

\subsubsection{Summary}

We summarize in Table~\ref{tab:OperationsonMSD} the asymptotic cost of basic operations discussed above and compare it with ST, through which the efficiency of $\CSD$ is established.

\begin{table*}[!htbp]
\setlength\extrarowheight{3pt}
\begin{center}
\resizebox{\linewidth}{!}{
\begin{tabular}{"M{9.6cm}|M{2.45cm}|M{2.5cm}|"}\thickhline
&\centering \large ST& \centering \large CSD  \tabularnewline \thickhline
\end{tabular}}

\vspace{0.1cm}

\resizebox{\linewidth}{!}{
\begin{tabular}{"p{9.6cm}|p{2.45cm}|p{2.5cm}|"}\thickhline
\ Storage&$\Theta(m\log (nt))$&$\mathcal{O}(\Psi n\log  (\Psi t))$\\ \thickhline
\end{tabular}}

\vspace{0.1cm}

\resizebox{\linewidth}{!}{
\begin{tabular}{"p{9.6cm}|p{2.45cm}|p{2.5cm}|"}\thickhline
\ Membership of a simplex $\sigma$&$\Theta(d_{\sigma}\log n)$ & $\mathcal{O}( d_{\sigma}\Gamma_0  \log {\Psi})$ \\\hline
\ Access Filtration Value& $\Theta(d_\sigma\log n)$ &$\mathcal{O}( d_{\sigma}\Psi  \log \Psi)$ \\\hline
\ Computing Filtration Value of Facets & $\mathcal{O}(d_\sigma^2\log n)$ &$\mathcal{O}(d_\sigma \Psi\log \Psi)$ \\\hline
\ Computing Filtration Value of Cofaces of codimension 1 & $\Theta(nd_\sigma\log n)$ &$\mathcal{O}(d\Psi\log \Psi)$ \\\thickhline
\end{tabular}}

\vspace{0.1cm}

\resizebox{\linewidth}{!}{
\begin{tabular}{"p{9.6cm}|p{2.45cm}|p{2.5cm}|"}\thickhline
\ Lazy Insertion of a simplex $\sigma$ &\ \ \ \ \ \ \ \ -- & $\mathcal{O}(d_\sigma \log \Psi)$\\\hline
\ Insertion of a simplex $\sigma$ &$\mathcal{O}(2^{d_{\sigma}}d_{\sigma}\log n)$ & $\mathcal{O}(\Psi d_\sigma^2\log \Psi)$\\\hline
\ Removal of a face & $\mathcal{O}(\Gamma_{d_\sigma}2^dd\log n)$ &$\mathcal{O}(\Psi^2 dd_\sigma\log \Psi)$\\\hline
\ Elementary Collapse & $\Theta(nd_\sigma\log n)$ &$\mathcal{O}(d_\sigma^2\Psi\log \Psi)$\\\thickhline
\end{tabular}}

\end{center}
\caption{Cost of performing basic operations on CSD in comparison with ST.}
\label{tab:OperationsonMSD}
\end{table*}

If the number of critical simplices is not large then $|\CSD|$ is smaller than $|\ST|$. The number of critical simplices is small unless we associate unique filtration values to a significant fraction of the simplices. 
As will be shown later, this assumption is relevant in applications.

We observe that while performing static queries, we pay a factor of $\Psi$ or $\Gamma_0$ in the case of CSD over the cost of the same operation in ST. In the case of dynamic operations we observe that the dependence on the dimension is exponentially smaller in CSD than in ST. Therefore, even if the number of critical simplices is polynomial in the dimension then, there is an exponential gap  between CSD and ST in both the storage and the efficiency of performing dynamic operations.

\subsection{Performance of CSD}\label{Performance}
CSD has been designed to store filtrations of simplicial complexes but it can be used to store simplicial complexes without a filtration. In this case,  $|M|=k$ and  CSD requires $\mathcal{O}(kd\log k)$ memory space, which matches the lower bound in Theorem~\ref{Lowerbound}, when $k=n^{\mathcal{O}(1)}$. In this case, CSD is very similar to SAL~\cite{BKT15}. Marc Glisse and Sivaprasad S.\ \cite{MarcSivaprasad} have performed experiments on SAL   and concluded that it is not only smaller in size but also faster than the Simplex Tree in performing insertion, removal, and edge contraction. 

CSD is also a  compact data structure to store filtrations, as its size matches (up to constant factors for small $d$) the lower bound of $\Omega(|M|(d\log n+\log t))$ in Theorem~\ref{CSDLowerBound}. 
Moreover, if $\Psi$ is small,  CSD is  not only a compact data structure since $|\CSD|$ is upper bounded by $\Psi n$, but, as shown in Table~\ref{tab:OperationsonMSD}, CSD is also a very efficient data structure as all basic operations depend polynomially on $d$ (as opposed to ST for which some operations depend exponentially on $d$).

As our analysis shows, we can express the complexity of CSD in terms of 
 a parameter $\Psi$ that reflects some ``local complexity'' of the simplicial complex. In the worst-case, 
 $\Psi=\Omega(m)$ as it can be observed in the complete complex with each simplex having a unique filtration value. However we conjecture that, even if $m$ is not small, $\Psi$ remains small for a large class of simplicial complexes of practical interest.  This conjecture is supported 
by the following experiment (and also Lemma~\ref{flagpsi}).

We considered a set of points  obtained by
sampling a Klein bottle in $\mathbb{R}^5$ and constructed its Rips
filtration (see Section~\ref{Flag} for definition) using libraries provided by the GUDHI project \cite{GUDHI}. We computed $\Gamma_0$ and $\Psi$ for  various  values of $t$. The resulting simplicial complex has $n=10,000$ vertices, dimension $d=17$ and has $m=10,508,486$ simplices of which $k=27,286$ are maximal. We record in Table~\ref{tab:Rips}, the values of $|\ST|$ and $|\CSD|$ for the various filtration ranges of the Rips complex constructed above. Figure~\ref{fig:Rips} shows a graphical illustration of the data.

\begin{table*}[!ht]\label{dataset1}
\begin{center}\resizebox{12cm}{!}{
\begin{tabular}{|c|c|c|c|c|c|c|c|c}\hline
No&$t$&$|\ST|$&$|\CSD|$&$\Gamma_0$&$\Gamma_0^{\text{avg}}$&$\Psi$&$\Psi^{\text{avg}}$
\\\hline
0&0&10,508,486&179,521&115&17.9&115&17.9
\\\hline
1&10&10,508,486&490,071&115&17.9&329&49.0
\\\hline
2&25&10,508,486&618,003&115&17.9&429&61.8
\\\hline
3&100&10,508,486&728,245&115&17.9&723&72.8
\\\hline
4&500&10,508,486&765,583&115&17.9&839&76.5
\\\hline
5&2,000&10,508,486&774,496&115&17.9&860&77.4
\\\hline
6&10,000&10,508,486&777,373&115&17.9&865&77.7
\\\hline
7&25,000&10,508,486&778,151&115&17.9&865&77.8
\\\hline
8&100,000&10,508,486&778,319&115&17.9&866&77.8
\\\hline
9&1,000,000&10,508,486&778,343&115&17.9&866&77.8
\\\hline
10&10,000,000&10,508,486&778,343&115&17.9&866&77.8
\\\hline
\end{tabular}}
\end{center}
\caption{Values of $|\CSD|$, $\Gamma_0$, $\Psi$, and $\Psi^{\text{avg}}$ for the simplicial complex generated from the above data set with increasing values of $t$. {Additionally, we provide $\Gamma_0^{\text{avg}}$ which is the expected number of maximal simplices that a vertex contains and $\Psi^{\text{avg}}$ which is the average size of $A_i$.}}
\label{tab:Rips}
\end{table*} 

\begin{figure*}[!t]
\centering
\resizebox{12cm}{!}{
\begin{tikzpicture}[-,shorten >=0.5pt,auto,node distance=2cm,
 thick,main node/.style={circle,fill=blue!10,draw,font=\sffamily\Huge\bfseries}]
\draw[-] (0,0.5) -- (-0.2,0.5);
\node at (-0.5,0.5) {\large{5.1}};
\foreach \j in {1,...,10}{
\pgfmathsetmacro\k{\j/5+5.1}
\draw[-] (0,0.5+\j) -- (-0.2,0.5+\j);
\node at (-0.5,0.5+\j) {\large{\k}};
}

\foreach \j in {0,...,9}{

\draw[-] (14,1.1*\j) -- (14.2,1.1*\j);
\node at (14.55,1.1*\j) {\large{\j00}};
}
\filldraw[white] (14.45,-0.15) rectangle (14.85,0.15);
\foreach \j in {0,...,7}{

\draw[-] (2*\j,0) -- (2*\j,-0.25);
\node at (2*\j,-0.45) {\large{\j}};
}

\draw[ultra thick, blue
]
 (0,1.25)-- (2,2.95)--  (2.8,3.45)--   (4,4.31) --  (5.4,4.41) --    (6.6,4.43) --   (8,4.44) --  (8.8,4.44)--  (10,4.45) --(12,4.45) --(14,4.45);  

\draw[ultra thick, red
]
 (0,1.1*1.15)-- (2,1.1*3.29)--  (2.8,1.1*4.29)--   (4,1.1*7.23) --  (5.4,1.1*8.39) --    (6.6,1.1*8.6) --   (8,1.1*8.65) --  (8.8,1.1*8.65)--  (10,1.1*8.67) --(12,1.1*8.67) --(14,1.1*8.67);  
 
\draw[ultra thick, green!90!black
]
(0,10.1)--(14,10.1);

 \draw[-,ultra thick] (0,0) rectangle (14,11);
\node at (10,10.5) {\Large\color{green!90!black}$\log_{10}|\ST|$};
\node at (7,4.8) {\Large\color{blue}$\log_{10}|\CSD|$};
\node at (3.3,7) {\Large\color{red}$\Psi$};

\node at (-1.65,10.1) {\large\rotatebox{90}{$\log_{10}|\ST|$}};
\node at (-1.65,3) {\large\rotatebox{90}{$\log_{10}|\CSD|$}};
\node at (15.43,5.5) {\Large$\Psi$};
\node at (7,-0.9) {\large $\log_{10}t$};

\draw[->]
(5.8,-1.2)--(8.2,-1.2);
\draw[->]
(-1.25,1.8)--(-1.25,4.2);
\draw[->]
(-1.25,8.9)--(-1.25,11.3);
\draw[->]
(15.2,4.9)--(15.2,6.1);


 \end{tikzpicture}
 }
 \caption{Values of $\log_{10}|\ST|$, $\log_{10}|\CSD|$, and $\Psi$ for the simplicial complex generated from above data set with increasing values of $t$. The blue curve corresponds to $\log_{10}|\CSD|$ on the left $y$ axis plotted against $\log_{10}t$ on the $x$-axis. The red curve corresponds to $\Psi$ on the right $y$ axis plotted against $\log_{10}t$ on the $x$-axis. The green curve(line) corresponds to $\log_{10}|\ST|$ on the left $y$ axis plotted against $\log_{10}t$ on the $x$-axis. }
\label{fig:Rips}

 \end{figure*}
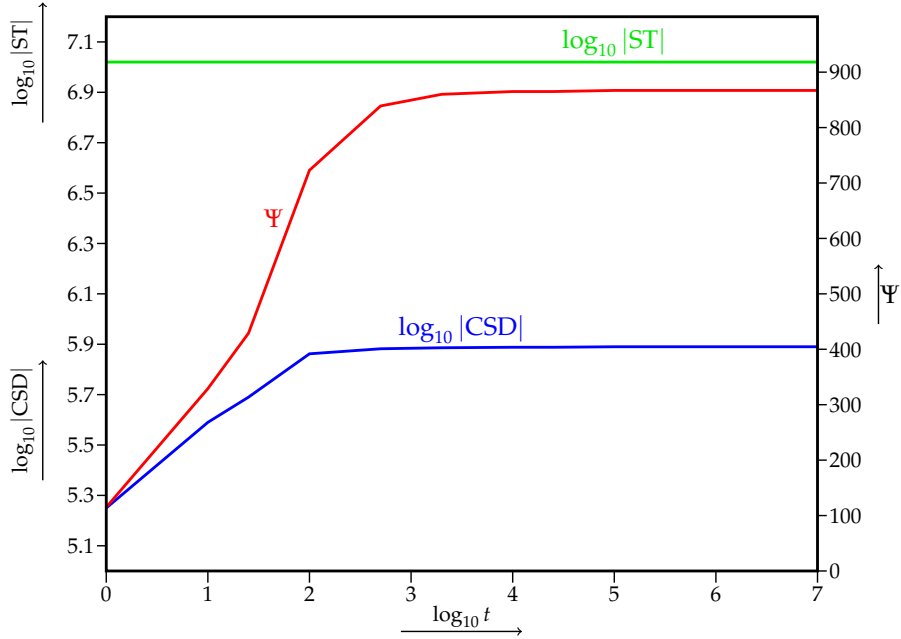
 
We note from Table~\ref{tab:Rips} that $\Gamma_0$ is significantly smaller than $k$, and also that $\Gamma_0^{\text{avg}}$  is  much smaller than $\Gamma_0$.
Also, from Figure~\ref{fig:Rips}, it is clear that there is a gap of an order of magnitude between $|\CSD|$ and $|\ST|$. Next, we note that $\Psi$ is remarkably smaller than $m$ (even notably smaller than $n$). This implies that  all operations can be efficiently implemented using CSD. More importantly, we remark that $\Psi^{\text{avg}}=|\CSD|/n$ is at most $77.8$ in the above experiment. Finally, we observe that despite increasing $t$ at a rapid rate, $|\CSD|$ grows very slowly after $t=100$. This is because the set of all possible filtration values of the Rips complex is small. Therefore, even for small values of $t$ the simplicial complex and its filtration is accurately captured by CSD.

\subsubsection*{Local Sensitivity of the Critical Simplex Diagram}
 It is worth noting that while the cost of basic operations are bounded using 
$\Gamma_0$ and $\Psi$, the actual cost is bounded by parameters such as ${\min}_i\left(|A_{\ell_i}^\star|\right)$, ${\min}_i\left(|A_{\ell_i}|\right)$,  and 
$Z_v$ (the set of all maximal simplices 
that contain the vertex $v$ which has been introduced in the Insertion paragraph in Section~\ref{sec:dynamicoperations}) to get a better estimate on the cost of 
these operations. These parameters are indeed local. To begin with, ${\min}_i\left(|A_{\ell_i}^\star|\right)$ captures the local information about a simplex 
$\sigma$ sharing a vertex with other maximal simplices of the 
complex. More precisely, it is the minimum, over all the vertices of 
$\sigma$, of the largest number of maximal simplices that contain the 
vertex. If $\sigma$ has a vertex which is contained in a few maximal 
simplices then, ${\min}_i\left(|A_{\ell_i}^\star|\right)$ is  small. Similarly, ${\min}_i\left(|A_{\ell_i}|\right)$ is the minimum, over all the vertices of 
$\sigma$, of the largest number of critical simplices that contain the 
vertex. This value depends not only on the structure of the filtration function but also on the filtration range.
Finally, $Z_v$ captures another local 
property of a simplex $\sigma$. Therefore, 
CSD is sensitive to the local structure of the complex. 

\subsubsection*{Robustness in Modification}

We now demonstrate the robustness of CSD, i.e., its ability to perform queries \emph{correctly} and \emph{efficiently} even when it might have stored redundant data such as simplices which are not critical or multiple copies of the same simplex with different filtration values. Consider modifying the filtration value of some simplex $\sigma\in K$ from $f(\sigma)$ to $s_{\sigma}(<f(\sigma))$. In the case of ST, we will have to modify the filtration value inside the node containing $\sigma$ and additionally check (and modify if needed) its faces in decreasing order of dimension. This requires time $\Theta(2^{d_\sigma}d_\sigma(\log n+\log t))$. However, in the case of CSD, we can lazy insert $\sigma$ into CSD in time $\mathcal{O}(d_\sigma\log \Psi)$, and the data structure is robust to such an insertion. This is because, all the operations can be performed correctly and with the same efficiency {up to constant additive factors in the worst case} after a lazy insertion. This is even true if some previously critical simplices need to be removed due to the lazy insertion of $\sigma$. 
For instance, consider the $\texttt{is\_critical}$ query on some simplex $\tau$. If  $\tau$ was  a face of $\sigma$ before modifying $f(\sigma)$ then, the \emph{minimal} filtration value of the nodes in $A_{\tau}$ correctly gives the filtration value of $\tau$ as $s_\sigma$ will now be one of the entries in $A_\tau$. Otherwise, if $\tau$ was not a face of $\sigma$, then the filtration value of $\tau$ remains unchanged, as the lazy insertion of $\sigma$ has not introduced a new simplex, but only a new filtration value to an existing simplex.
Therefore, we can think of using the data structure to manipulate simplicial complexes in very short time through a collection of lazy insertions and perform a clean-up operation at the end of the collection of lazy insertions, or even think of performing the clean-up operation in parallel to the lazy insertions.
We remark here that if we lazy insert $r$ simplices then in the worst case, $\Psi$ grows to $r+\Psi$. In other words,  the presence of redundant simplices, implies that the efficiency will now depend on $r+\Psi$ instead of $\Psi$, but the redundancy will not affect the correctness of the operations.


\subsection{{A Sequence of  Representations for Simplicial Complexes and their Filtrations}}

Boissonnat et al.\ \cite{BKT15} in their paper on Simplex Array List described a sequence of data
structures, each more powerful than the previous ones (but also
bulkier). In that sequence of data structures $\langle \Lambda\rangle$, we had $\Lambda_{i}=i\mhyphen \MSD$ ($\MSD$ referred to earlier in this paper is equal to $1\mhyphen\MSD$).
Furthermore, they note that in the $i^{\textrm{th}}$ element of the sequence,
every node which is not a leaf (sink) in the data structure corresponds to a unique $i$-simplex
in the simplicial complex. Also for all $i\mhyphen\MSD$, $i\in\mathbb{N}$,
they state that it is a NFA recognizing all the simplices in the
complex. As one moves along the sequence, the size of the data structure
blows up by a factor of $d$ at each step. But in return, there is a gain in the efficiency of
searching for simplices as the membership query depends on
$\Gamma_i$ which decreases as $i$ increases.

We note here that $\CSD$ described in this paper is exactly the same as $0\mhyphen\MSD$, when $t=0$ (we ignore the structure of the connected component, which is a path in SAL but a star in CSD). Therefore, $\CSD$ supercedes $0\mhyphen\MSD$. There is no change in representation of a simplex between SAL and CSD; instead we only store more simplices (i.e, all critical simplices) in CSD. Therefore, in the same vein as $\langle\Lambda\rangle$, we can define a sequence of data
structures, each more powerful than the previous ones (but also
bulkier). More formally, consider the sequence of data structures
$\langle \Pi\rangle$, where $\Pi_{0}=\CSD$ and $\Pi_i =
\underset{j=0}{\overset{t}{\bigcup}}i\mhyphen\MSD(M_j)$, for all
$i\in\mathbb{N}$. Furthermore, for all $i\in\mathbb{N}$, we will refer to the data structure $\Pi_i$ by the name $i\mhyphen \CSD$ (we will continue to refer to $0\mhyphen\CSD$ as $\CSD$). 
As we move along the sequence $\Pi$, the size of the data structure
blows up by a factor of $d$ at each step. But in return, we gain efficiency in
searching for simplices as the membership query depends on
$\Gamma_i$ which decreases as $i$ increases. Additionally, we gain efficiency in accessing filtration value of a simplex as the complexity no longer depends on $\Psi=\Psi_0$ but on a smaller parameter, $\Psi_i$, which is the maximum number of critical cofaces that any $i$-simplex can have in the complex.

Marc Glisse and Sivaprasad S.\ implemented SAL \cite{MarcSivaprasad} for Data Set mentioned in Section 3.4, and then performed insertion and removal of random simplices, and contracted randomly chosen edges. They observed that 1-SAL outperformed 0-SAL in low dimensions. However, 0-SAL performed better than 1-SAL in higher dimensions.
Therefore, in a similar vein, it would be worth exploring for which class of simplicial complexes, $i\mhyphen$CSD is the best data structure in the CSD family (for every $i\in\mathbb{N}$). 

\section{Construction of Flag Complexes}\label{Flag}
Recall that the flag complex of an undirected graph $G$ is defined as an abstract simplicial complex, whose simplices are the sets of vertices in the cliques of $G$. Let $(P, \|\cdot\|)$ be a metric space where $P$ is a discrete point-set. Given a positive real number $r > 0$, the Rips complex is the abstract simplicial complex $\mathcal{R}^r(P)$ where a simplex  $\sigma\in\mathcal{R}^r(P)$ if and only if $\|p-q\|\le 2r$ for every pair of vertices of $\sigma$. Note that the Rips complex is a special case of a flag complex. 
Rips filtrations are widely used in Topological Data Analysis since they are easy to compute and they allow to robustly reconstruct the homology of a sample shape via the computation of its persistence diagram~\cite{CCSGGO09}.

In this section, we will only consider a specific filtration for flag complexes which is of significant interest as it includes the Rips filtration. The filtration value of a vertex is $0$. The filtration value of every edge $e$ in the complex, denoted $f(e)$, is given as part of the input. The filtration value of a simplex of higher dimension is equal to the maximum of the filtration values of all the edges in the simplex. 
\subsection{Edge-Deletion Algorithm for Construction of Flag Complexes}\label{Algoconstruction}

Let $G$ be the (weighted) graph of the simplicial complex $K$. Let $\Delta$ denote the maximum degree of the vertices of $G$. To represent $K$ using ST, Boissonnat and Maria \cite{SimplexTree} propose to compute and insert the $\ell$-skeleton of $K$ into ST and to incrementally increase $\ell$ from $1$ to $d$. Therefore, the time to construct  the ST representing the flag complex is $\mathcal{O}(md\log n)$.

To represent $K$ using CSD, we propose an \emph{edge-deletion} algorithm, which is significantly faster than the construction algorithm for ST. We recall that in Section \ref{sec:notations}\ we defined $M_h$ to denote the set of critical simplices in the complex with filtration value $h$. For the following algorithm, we will assume that all the edge weights are distinct.

\paragraph{Preprocessing Step}  We first compute all the maximal cliques in $G$ in time $\mathcal{O}(k\cdot n^\omega)$ \cite{MU04}, where $\omega<2.38$ \cite{L14} is the matrix multiplication exponent, i.e., $n^\omega$ is the time needed to multiply two $n\times n $ matrices. 
We store these maximal simplices in a Prefix Tree (like MxST of \cite{BKT15}). The filtration value given to the edges provides a natural ordering to the edges of the complex. We consider edges in descending order of their filtration value. Let $e_i$ be the edge with the $i^{\text{th}}$ highest filtration value.
Set $i=1$.

\paragraph{Step 1} In this step, we would like to compute $M_{f(e_i)}$ in order to build the CSD of the complex. 
We observe that a clique in $G$ containing the edge $e_i$ is maximal if and only if the vertices of the clique form  a critical simplex in  $M_{f(e_i)}$. In the above observation we used the assumption that all the edge weights in the graph are distinct.
We list all the maximal simplices containing the edge $e_i$ in time $\mathcal{O}\left(\left|M_{f(e_i)}\right|\Delta^\omega\right)$, by using the algorithm presented by  Makino and Uno \cite{MU04} on a subgraph of $G$ in the following way. 
We build an induced subgraph $H$ of $G$ that contains the vertices of edge $e_i$ and all the vertices which are adjacent to \emph{both} vertices of $e_i$. We note that every maximal clique in $H$ is a maximal clique in $G$ containing edge $e_i$, and vice versa. Therefore, if we run Makino and Uno's algorithm on $H$ (which contains at most $\Delta+1$ vertices), we obtain all the maximal cliques in $G$ containing edge $e_i$.

\paragraph{Step 2} Next, we recognize the maximal simplices of $K$ in $M_{f(e_i)}$ in time $\mathcal{O}\left(\left|M_{f(e_i)}\right|d\log n\right)$ by checking each simplex $\sigma$ in $M_{f(e_i)}$ with the Prefix tree built in the preprocessing step in time $\mathcal{O}(d_\sigma \log n)$ per simplex. We remark here that all the simplices in $M_{f(e_1)}$  are maximal simplices in $K$, since $e_1$ has the largest filtration value. 

\paragraph{Step 3} We perform lazy insertion of the simplices in $M_{f(e_i)}$ into the CSD and since we have identified the maximal simplices in $M_{f(e_i)}$, we know whether to insert them in $A_j^\star$ or not, for each $A_j$. This takes time $\mathcal{O}\left(\left|M_{f(e_i)}\right|d\log \Psi\right) =\mathcal{O}\left(\left|M_{f(e_i)}\right|d^2\log n\right)$.

\paragraph{Step 4} Finally, we remove $e_i$ from $G$, increment $i$ by 1, and repeat the procedure from step 1 until $G$ has no edges left.  

This entire construction takes time $\mathcal{O}(|M|(\Delta^\omega+d^2\log n))=\mathcal{O}\left(|M|n^{2.38}\right)$, which is significantly better than that of constructing a representation of $K$ by ST 
as $|M|$ can be considerably (exponentially) smaller than $m$.
Finally, we remark that if we lose the distinctness assumption on the edge weights, then some of the maximal cliques of $H$ listed in Step~1 
may not correspond to critical simplices, and thus we will need to perform the operation \texttt{is\_critical} after constructing the CSD and remove the star trees not corresponding to any of the critical simplices. 
\section{Construction of Relaxed Delaunay Complexes}
\label{Delaunay}
Let $Q$ be a finite subset of a metric space $(P, \|\cdot\|)$ where $P$ is a discrete point-set. Given a relaxation parameter $\rho\ge 0$, we define the notion of being `witnessed' as follows. A simplex $\sigma = \{q_0,\ldots , q_{d_\sigma}\}\subseteq Q$ belongs to 
$\text{Del}^\rho(Q, P)$ \cite{D08,BDG15} if and only if there exists $x\in P$ {that strongly $\rho$-witnesses $\sigma$,} i.e. such that for all $q_i\in \sigma$, and for all $q\in Q$ the following holds:
\begin{align*}
\|x-q_i\| \leq \|x-q\| + \rho.
\end{align*}

The parameter $\rho$ defines a filtration, 
which has been used in topological data analysis. More explicitly, the filtration value of a simplex $\sigma$ in $\text{Del}^\rho(Q, P)$  is the smallest $\rho^\prime\le \rho$, such that $\sigma$ is in $\text{Del}^{\rho^\prime}(Q, P)$. 
For this entire section, we assume that the filtration range is $\llbracket t\rrbracket$ (obtained after appropriate scaling).

We define a matrix $D$ of size $|P|\times |Q|$ as follows. For every $x\in P$ and $\ell\in [|Q|]$ let $D(x,\ell)$ denote the $\ell^{\text{th}}$ nearest neighbor of $x$ in $Q$ (ties are broken arbitrarily). For every $x\in P$, $i\in\llbracket t\rrbracket$, let $\ell_x^i$ be the largest integer such that $\left|\|x-D(x,1)\|-\|x-D(x,\ell_x^i)\|\right|\le \rho i/t$. Let $\sigma_x^i=\{D(x,1),D(x,2),\ldots ,D(x,\ell_x^i)\}$ and  let $W=\left\{\sigma_x^i\bigl\vert x\in P, i\in \llbracket t \rrbracket\right\}$.
We note that 
the set  $M$  of critical simplices in the complex is contained in $W$.

\begin{lemma}\label{DelaunayStrong}
$M\subseteq W.$
\end{lemma}
\begin{proof}
Let $\tau=\{v_0,\ldots v_{d_\tau}\}$ be a critical simplex in $M$ with filtration value $f(\tau)$. By definition of $\text{Del}^\rho(Q, P)$, we have that there exists a point $x\in P$ which $\left(f(\tau)\rho/t\right)$-strongly witnesses  $\tau$. 
Since $\tau$ is critical, we have that for every $q\in Q\setminus \tau$, the following holds:
$$\forall i\in\llbracket d_\tau\rrbracket,\  \left|\|x-v_i\|-\|x-q\|\right|> \rho f(\tau)/t.$$
Therefore, we have that for every $i\in [d_\tau+1]$, $D(x,\ell_x^i)\in \tau$, or more precisely, $\sigma_x^{d_\tau+1}=\tau$.
\end{proof}

The above lemma provides a characterization of $\text{Del}^\rho(Q, P)$: it can have at most $|P|(d+1)$ critical simplices (where $d$ is the dimension of $\text{Del}^\rho(Q, P)$). We note here that typically $P$ is a relatively small set. For example, in the experiments performed by Boissonnat and Maria (Table 1 of \cite{SimplexTree}), we note that the cardinality of $P$ is about a few ten thousands while the number of simplices in the complex is over a hundred million. Therefore, this provides practical evidence of the compact representation of $\text{Del}^\rho(Q, P)$ through CSD.

Under the assumption that for any $x,\ell$,  $D(x,\ell)$ could be computed in $\mathcal{O}(1)$ time (i.e., $D$ is computed as part of the preprocessing), Boissonnat and Maria \cite{SimplexTree} described an algorithm to construct the ST representation of the relaxed witness complex. Their algorithm can be easily adapted to construct $\text{Del}_w^\rho(Q, P)$ in time $\mathcal{O}(tmd\log n)$.

In the case of CSD, we propose a new \emph{matrix-parsing} algorithm which builds $\text{Del}^\rho(Q, P)$ in time $\mathcal{O}(|P|d^2\log \Psi)$ (assuming an oracle to access $D$). 
It is easy to see that all the simplices in $W$ can be constructed in $\mathcal{O}(|W|d\log n)=\mathcal{O}(|P|d^2\log n)$ time by sequentially computing the simplices $\sigma_x^i$ for all the $x\in P$, i.e., by parsing the matrix $D$ one row at a time.  We lazy insert all the simplices in $W$ to the CSD in time $\mathcal{O}(|P|d^2\log \Psi)$. We finish the construction by performing a clean-up operation to remove the redundant simplices (i.e., non-critical simplices) that were inserted. 

\section{Conclusion}\label{Conclusion}
In this paper, we introduce a new data structure called the Critical Simplex Diagram (CSD) to represent filtrations of simplicial complexes. In this data structure, we store only those simplices which are critical with respect to the filtration value, i.e., we store a simplex if and only if all its cofaces are of a (strictly) higher filtration value than the filtration value of the simplex itself. We then show how to efficiently perform basic operations on simplicial complexes by only storing these (critical) simplices. This is summarized in Table~\ref{tab:OperationsonMSD}. Finally, we showed how to (quickly) construct the CSD representation of flag complexes and relaxed Delaunay complexes. 

{As a future direction of research, we would like to develop algorithms for computing persistent homology of a filtration using the CSD representation.} This is a very important question, as computing persistence is at the heart of topological data analysis. On the other hand, it is open to obtain better bounds on $\Psi$ and $\Gamma_i$ for specific  complexes such as the Rips complex or the relaxed Delaunay complex  by assuming some notion of geometric regularity. Also, it would be interesting to obtain lower bounds on the various query times (such as membership, insertion/removal), by assuming an optimal storage of $\mathcal{O}(\kappa d\log n)$ ($\kappa=|M|$ is the number of critical simplices).  From the standpoint of practice, we would like to find fast construction algorithms under the CSD representation for other simplicial complexes of interest such as the alpha complex and the relaxed witness complex.
Finally, we would like to implement this data structure and check its performance versus the Simplex Tree in practice.

\subsection*{Acknowledgements}
We would like to thank the anonymous reviewers whose comments helped us improve the presentation of the paper. 

\bibliographystyle{alpha}
\bibliography{References}

\newcommand{\etalchar}[1]{$^{#1}$}
\begin{thebibliography}{CCG{\etalchar{+}}09}

\bibitem[ALS12]{DataStructure3}
Dominique Attali, Andr{\'{e}} Lieutier, and David Salinas.
\newblock Efficient data structure for representing and simplifying simplicial
  complexes in high dimensions.
\newblock {\em Int. J. Comput. Geometry Appl.}, 22(4):279--304, 2012.

\bibitem[BDG15]{BDG15}
Jean{-}Daniel Boissonnat, Ramsay Dyer, and Arijit Ghosh.
\newblock A probabilistic approach to reducing algebraic complexity of delaunay
  triangulations.
\newblock In {\em Algorithms - {ESA} 2015 - 23rd Annual European Symposium,
  Patras, Greece, September 14-16, 2015, Proceedings}, pages 595--606, 2015.

\bibitem[BKT17]{BKT15}
Jean{-}Daniel Boissonnat, {Karthik {C. S.}}, and S{\'{e}}bastien Tavenas.
\newblock Building efficient and compact data structures for simplicial
  complexes.
\newblock {\em Algorithmica}, 79(2):530--567, 2017.

\bibitem[BM14]{SimplexTree}
Jean{-}Daniel Boissonnat and Cl{\'{e}}ment Maria.
\newblock The simplex tree: An efficient data structure for general simplicial
  complexes.
\newblock {\em Algorithmica}, 70(3):406--427, 2014.

\bibitem[CCG{\etalchar{+}}09]{CCSGGO09}
Fr{\'{e}}d{\'{e}}ric Chazal, David Cohen{-}Steiner, Marc Glisse, Leonidas~J.
  Guibas, and Steve Oudot.
\newblock Proximity of persistence modules and their diagrams.
\newblock In {\em Proceedings of the 25th {ACM} Symposium on Computational
  Geometry, Aarhus, Denmark, June 8-10, 2009}, pages 237--246, 2009.

\bibitem[DFW14]{DFW14}
Tamal~K. Dey, Fengtao Fan, and Yusu Wang.
\newblock Computing topological persistence for simplicial maps.
\newblock In {\em 30th Annual Symposium on Computational Geometry, SOCG'14,
  Kyoto, Japan, June 08 - 11, 2014}, page 345, 2014.

\bibitem[dS08]{D08}
Vin de~Silva.
\newblock A weak characterisation of the delaunay triangulation.
\newblock {\em Geometriae Dedicata}, 135(1):39--64, 2008.

\bibitem[Dur03]{D03}
Guillermo Dur\'an.
\newblock Some new results on circle graphs.
\newblock {\em Matem\'atica Contempor\^anea}, 2003.

\bibitem[EH10]{EH10}
Herbert Edelsbrunner and John Harer.
\newblock {\em Computational Topology - an Introduction}.
\newblock American Mathematical Society, 2010.

\bibitem[ELS10]{ELS10}
David Eppstein, Maarten L{\"{o}}ffler, and Darren Strash.
\newblock Listing all maximal cliques in sparse graphs in near-optimal time.
\newblock In {\em Algorithms and Computation - 21st International Symposium,
  {ISAAC} 2010, Jeju Island, Korea, December 15-17, 2010, Proceedings, Part
  {I}}, pages 403--414, 2010.

\bibitem[Gal14]{L14}
Fran{\c{c}}ois~Le Gall.
\newblock Powers of tensors and fast matrix multiplication.
\newblock In {\em International Symposium on Symbolic and Algebraic
  Computation, {ISSAC}'14, Kobe, Japan, July 23-25, 2014}, pages 296--303,
  2014.

\bibitem[Gav74]{G74}
Fanica Gavril.
\newblock Algorithms on circular-arc graphs.
\newblock {\em Networks}, 4(4):357--369, 1974.

\bibitem[GKS13]{GKS13}
Martin Grohe, Stephan Kreutzer, and Sebastian Siebertz.
\newblock Characterisations of nowhere dense graphs (invited talk).
\newblock In {\em {IARCS} Annual Conference on Foundations of Software
  Technology and Theoretical Computer Science, {FSTTCS} 2013, December 12-14,
  2013, Guwahati, India}, pages 21--40, 2013.

\bibitem[Gol80]{G80}
Martin~Charles Golumbic.
\newblock {\em Algorithmic graph theory and perfect graphs}.
\newblock Computer science and applied mathematics. Academic Press, New York,
  1980.

\bibitem[GS]{MarcSivaprasad}
M.~Glisse and S.~Sivaprasad.
\newblock Private communication.

\bibitem[Hed73]{H73}
Stephen~T. Hedetniemi.
\newblock Hereditary properties of graphs.
\newblock {\em Journal of Combinatorial Theory, Series B}, 14(1):94 -- 99,
  1973.

\bibitem[MU04]{MU04}
Kazuhisa Makino and Takeaki Uno.
\newblock New algorithms for enumerating all maximal cliques.
\newblock In {\em Algorithm Theory - {SWAT} 2004, 9th Scandinavian Workshop on
  Algorithm Theory, Humlebaek, Denmark, July 8-10, 2004, Proceedings}, pages
  260--272, 2004.

\bibitem[Pri95]{P95}
Erich Prisner.
\newblock Graphs with few cliques.
\newblock In {\em 7th Quadrennial International Conference on the Theory and
  Applications of Graphs, Graph Theory, Combinatorics, and Applications}, pages
  945--956, 1995.

\bibitem[Pro15]{GUDHI}
The~GUDHI Project.
\newblock {\em GUDHI User and Reference Manual}.
\newblock GUDHI Editorial Board, 2015.

\bibitem[RS07]{RS07}
Bill Rosgen and Lorna Stewart.
\newblock Complexity results on graphs with few cliques.
\newblock {\em Discrete Mathematics {\&} Theoretical Computer Science}, 9(1),
  2007.

\bibitem[Spi03]{S03}
J.P. Spinrad.
\newblock {\em Efficient Graph Representations.: The Fields Institute for
  Research in Mathematical Sciences.}
\newblock Fields Institute monographs. American Mathematical Soc., 2003.

\end{thebibliography}
\end{document}